\documentclass[letterpaper,onecolumn,11pt,accepted=2024-04-08]{quantumarticle}
\pdfoutput=1
\usepackage[utf8]{inputenc}
\usepackage[english]{babel}
\usepackage[T1]{fontenc}
\usepackage{amsmath}

\usepackage{tikz}
\usepackage{lipsum}

\usepackage{soul,times,amsthm,amsfonts,bbm,mathrsfs,xcolor,subfigure,graphicx,enumitem,booktabs,euscript,mathtools,mathabx,eczoo,nicefrac}

\allowdisplaybreaks
\usepackage{graphicx}
\usepackage[unicode]{hyperref}
\hypersetup{
    colorlinks,
    linkcolor={blue!80!black},
    citecolor={blue!80!black},
    urlcolor={blue!80!black}
}
\usepackage[normalem]{ulem}
\newcommand\redout{\bgroup\markoverwith{\textcolor{red}{\rule[0.5ex]{2pt}{0.8pt}}}\ULon}
\usepackage[font=footnotesize]{caption}
\DeclarePairedDelimiter\bra{\langle}{\rvert}
\DeclarePairedDelimiter\ket{\lvert}{\rangle}
\DeclarePairedDelimiterX\braket[2]{\langle}{\rangle}{#1\,\delimsize\vert\,\mathopen{}#2}

\newtheorem{theorem}{Theorem}[section]
\newtheorem{lemma}[theorem]{Lemma}

\newtheorem{proposition}[theorem]{Proposition}

\newtheorem{construction}{Construction}[section]
\theoremstyle{remark}

\newtheorem{definition}{Definition}[section]
\newtheorem{example}{Example}

\newcommand\nc\newcommand
\nc\bfa{{\boldsymbol a}}\nc\bfA{{\boldsymbol A}}\nc\cA{{\EuScript A}}
\nc\bfb{{\boldsymbol b}}\nc\bfB{{\boldsymbol B}}\nc\cB{{\EuScript B}}
\nc\bfc{{\boldsymbol c}}\nc\bfC{{\boldsymbol C}}\nc\cC{{\mathscr C}}
\nc\bfd{{\boldsymbol d}}\nc\bfD{{\boldsymbol D}}\nc\cD{{\mathscr D}}
\nc\bfe{{\boldsymbol e}}\nc\bfE{{\boldsymbol E}}\nc\cE{{\EuScript E}}
\nc\bff{{\boldsymbol f}}\nc\bfF{{\boldsymbol F}}\nc\cF{{\mathscr F}}
\nc\bfg{{\boldsymbol g}}\nc\bfG{{\boldsymbol G}}\nc\cG{{\EuScript G}}
\nc\bfh{{\boldsymbol h}}\nc\bfH{{\boldsymbol H}}\nc\cH{{\mathcal H}}
\nc\bfi{{\boldsymbol i}}\nc\bfI{{\boldsymbol I}}\nc\cI{{\mathcal I}}
\nc\bfj{{\boldsymbol j}}\nc\bfJ{{\boldsymbol J}}\nc\cJ{{\EuScript J}}
\nc\bfk{{\boldsymbol k}}\nc\bfK{{\boldsymbol K}}\nc\cK{{\EuScript K}}
\nc\bfl{{\boldsymbol l}}\nc\bfL{{\boldsymbol L}}\nc\cL{{\EuScript L}}
\nc\bfm{{\boldsymbol m}}\nc\bfM{{\boldsymbol M}}\nc\cM{{\EuScript M}}
\nc\bfn{{\boldsymbol n}}\nc\bfN{{\boldsymbol N}}\nc\cN{{\EuScript N}}
\nc\bfo{{\boldsymbol o}}\nc\bfO{{\boldsymbol O}}\nc\cO{{\EuScript O}}
\nc\bfp{{\boldsymbol p}}\nc\bfP{{\boldsymbol P}}\nc\cP{{\EuScript P}}
\nc\bfq{{\boldsymbol q}}\nc\bfQ{{\boldsymbol Q}}\nc\cQ{{\mathcal Q}}
\nc\bfr{{\boldsymbol r}}\nc\bfR{{\boldsymbol R}}\nc\cR{{\EuScript R}}
\nc\bfs{{\boldsymbol s}}\nc\bfS{{\boldsymbol S}}\nc\cS{{\EuScript S}}
\nc\bft{{\boldsymbol t}}\nc\bfT{{\boldsymbol T}}\nc\cT{{\EuScript T}}
\nc\bfu{{\boldsymbol u}}\nc\bfU{{\boldsymbol U}}\nc\cU{{\EuScript U}}
\nc\bfv{{\boldsymbol v}}\nc\bfV{{\boldsymbol V}}\nc\cV{{\mathscr V}}
\nc\bfw{{\boldsymbol w}}\nc\bfW{{\boldsymbol W}}\nc\cW{{\mathscr W}}
\nc\bfx{{\boldsymbol x}}\nc\bfX{{\boldsymbol X}}\nc\cX{{\EuScript X}}
\nc\bfy{{\boldsymbol y}}\nc\bfY{{\boldsymbol Y}}\nc\cY{{\mathscr Y}}
\nc\bfz{{\boldsymbol z}}\nc\bfZ{{\boldsymbol Z}}\nc\cZ{{\EuScript Z}}
\nc{\1}{{\mathbbm{1}}}

\nc{\remove}[1]{}

\DeclareSymbolFont{bbold}{U}{bbold}{m}{n}
\DeclareSymbolFontAlphabet{\mathbbold}{bbold}

\DeclareMathOperator{\supp}{supp}

\DeclareMathOperator{\Tr}{Tr}

\newcommand{\new}[1]{\textcolor{black}{#1}}

\newcommand{\E}{{\mathbb E}}
\nc\reals{{\mathbb R}}
\nc{\ff}{{\mathbb F}}
\nc{\PP}{{\mathbb P}}
\nc{\complex}{{\mathbb C}}

\newcommand{\krA}{\mathfrak{K}_{\mathcal{A}}}

\DeclarePairedDelimiter\ceil{\lceil}{\rceil}

\begin{document}

\title{A family of permutationally invariant quantum codes}

\author{Arda Aydin}
\affiliation{Department of ECE and Institute for Systems Research, University of Maryland, College Park, MD 20742}
\email{aaydin@umd.edu}
\author{Max A. Alekseyev}
\affiliation{Department of Mathematics, The George Washington University, Washington, DC 20052}
\author{Alexander Barg}
\affiliation{Department of ECE and Institute for Systems Research, University of Maryland, College Park, MD 20742}
\email{abarg@umd.edu}
\orcid{0000-0003-1985-4623}
\maketitle

\begin{abstract}
 We construct a new family of permutationally invariant codes that correct $t$ Pauli errors for any $t\ge 1$.  We also show that codes in the new family correct quantum deletion errors as well as spontaneous decay errors. \new{Our construction contains some of the previously known permutationally invariant quantum codes as particular cases, which also admit transversal gates}. In many cases, the codes in the new family are shorter than the best previously known explicit permutationally invariant codes for Pauli errors and deletions. \new{Furthermore, our new code family includes a new $((4,2,2))$ optimal single-deletion-correcting code}.  As a separate result, we generalize the conditions for permutationally invariant codes to correct $t$ Pauli errors from the previously known results for $t=1$ to any number of errors. For small $t$, these conditions can be used to construct new examples of codes by computer.
\end{abstract}

\section{Introduction}

Quantum error correction is one of the essential components of quantum computing that aims to protect quantum information from errors caused by quantum noise, such as decoherence.  Mapping a quantum state to be protected into a higher-dimensional Hilbert space of the physical system is a significant part of quantum error correction. A subspace of the Hilbert space of a physical system is called a quantum code for a given type of errors if it satisfies certain specific conditions for error correction \cite{knill}. In many applications, it is desirable to construct quantum codes that lie within the ground space of the system. Motivated by this goal, in this paper, we study permutation-invariant quantum codes whose codewords form ground states of the ferromagnetic Heisenberg model.

Recall that Heisenberg's model characterizes interactions between spins in the system. Two spin-$1/2$ particles are coupled by an interaction described by the Hamiltonian $ \hat{H} \propto\;\bfS_i\bfS_j,$ where $ \bfS_i $ and $ \bfS_j $ are spin operators for the particles $ i $ and $ j $, respectively. In the absence of an external magnetic field, the Heisenberg ferromagnetic model is described by the Hamiltonian that can be written in the form
\begin{align*}
    \hat{H}=-2\sum\limits_{i<j}J_{ij}\bfS_i\bfS_j,
\end{align*}
where $ J_{ij} $ is the exchange (coupling) 
constant between particles $i$ and $j$ in the system. Note that
$ J_{ij} > 0 $ since we are considering the ferromagnetic model. See \cite[Ch.~1,4]{blundell} for a detailed discussion of this model. 
It can be shown that 
    \begin{align*}
        \bfP_{ij} = \frac{1}{2}{\bfI} + 2\bfS_i\bfS_j,
    \end{align*}
where $ \bfI $ is the identity and where $\bfP_{ij}$ is the swap operator that exchanges spin $ i $
and spin $j$,
essentially swapping the spin-$\frac12$ particles $i$ and $j$, e.g., $ \bfP_{12}\ket{\updownarrows\upuparrows} = \ket{\downuparrows\upuparrows}$. The Hamiltonian of the Heisenberg model can be written in terms of the swap operators in the following way:
\begin{align*}
     \hat{H}=-\sum\limits_{i<j}J_{ij}\left(\bfP_{ij}- \frac{1}{2}{\bfI}\right).
\end{align*}
A state $\ket\psi$ is called {\em permutation-invariant} if it is preserved by all swap operators $\bfP_{ij},$ i.e., $\ket\psi$ is a common eigenstate of the swap operators with eigenvalue 1.
Denoting $ J=\sum_{i<j} J_{ij}$, we observe that for any permutation-invariant state $ \ket{\psi}, $ 
\begin{align*}
    \left(\hat{H}-\frac{J}{2}{\bfI}\right)\ket{\psi} = -\sum_{i<j} J_{ij}\bfP_{ij}\ket{\psi}=-J\ket{\psi}.
\end{align*}
Since $ J_{ij}>0 $, the spectral norm of $ \hat{H}-\frac{J}{2}{\bfI} $ is bounded above by $ J $,  
so the smallest eigenvalue of the Hamiltonian is $ -J/2 ,$ and its corresponding eigenstate is $ \ket{\psi} $. Therefore, any permutation-invariant state is a ground state in the ferromagnetic Heisenberg model \cite{ouyangPI}.

Permutation-invariant codes were introduced in the works of Ruskai and Pollatsek \cite{ruskaiExchange,ruskai-polatsek}. The
codes they constructed encode a single logical qubit, and are capable of correcting all one-qubit errors and certain types of two-qubit errors. In particular, Ruskai's $((9,2,3))$ code \cite{ruskaiExchange} has the basis
   \begin{align*}
   \ket {0_L}&=\ket{0^9}+\frac1{\sqrt{28}}\sum_{\pi}\ket{1^60^3}\\
   \ket {1_L}&=\ket{1^9}+\frac1{\sqrt{28}}\sum_{\pi}\ket{0^61^3},
   \end{align*}
where the sum is extended to all permutations of the argument state. This code is obtained as a symmetrized
version of Shor's \eczoo[{9-qubit code}]{shor_nine} \cite{ShorCode}. Generalizing this construction, Ouyang \cite{ouyangPI} found a family of permutation-invariant codes that correct $t$ 
arbitrary errors and $t$ spontaneous decay errors. The family is parameterized by integers
$g$, $n$, and $u$ (hence the name ``gnu codes"), and the shortest $t$-error-correcting codes in it are of length $(2t+1)^2$. Ouyang subsequently showed that permutation-invariant codes are capable of supporting reliable quantum storage, quantum sensing, and decoherence-free communication
\cite{ouyangSensors, ouyangStorage, ouyanggCommunication}. 

In hindsight, it is clear that permutation-invariant codes also support recovery of encoded states from deletion errors: since
permutations preserve the states, deletion of arbitrary $t$ positions is not different from deleting the {\em first
$t$ qubits}, and thus deletions are equivalent to erasures. However, making this idea formal requires a rigorous
definition of the quantum deletion channel as well as proving the equivalence. This was accomplished in the works of Nakayama and Hagiwara \cite{nakayamaFirst,hagiwara4qubit,hagiwaraSingle}, who also observed that permutation-invariant codes are capable of correcting deletion errors and constructed single-deletion-correcting codes. Subsequently, works \cite{hagiwaraDeletion,ouyangDeletion} showed that Ouyang's {\em gnu} codes can correct $t$ deletions. In particular, the shortest known code to correct $t$ deletions comes from this family, and it has length $(t+1)^2.$ 

Recall that correcting deletions has a long history in classical coding theory. This problem was introduced as far back as 1965 by Levenshtein \cite{levenshtein1966binary}, and it has been studied both in combinatorial and probabilistic versions. Despite a number of spectacular advances in recent years, both problems are still far from solution: for instance, the only
case when the optimal codes are known is correcting a single deletion \cite{varshamov1965codes}. We refer the reader to two recent surveys dealing with constructive and capacity aspects of transmission over the deletion channel, \cite{haeupler2021synchronization} and \cite{cheraghchi2020overview}, published in the special issue devoted to the scientific legacy of Vladimir Levenshtein.

In quantum coding theory, a deletion can be modeled as a partial trace operation where the traced-out qubits are unknown. It turns out that the performance of permutation-invariant codes for correcting errors or deletions is sometimes amenable to analysis. Focusing on this code family, we study the error correction (Knill--Laflamme) conditions for general permutation-invariant codes. Using deletion correction as motivation, we propose a new family of permutation-invariant codes defined by their parameters $g$, $m$, $\delta$, and $ \epsilon $. The shortest codes in this family have length $(2t+1)^2-2t$ and can correct all $t$ patterns of qubit errors and $ 2t $ deletion errors. The shortest $t$-error-correcting permutation-invariant codes known previously are due to Ouyang and require $2t$ more physical qubits than the codes that we propose. Specializing our construction to $t=1,$ we observe that the length of our code is the same as the Pollatsek--Ruskai's $((7,2,3))$ permutation-invariant code \cite{ruskai-polatsek}, although the two codes are different. The authors of \cite{ruskai-polatsek} also derived explicit conditions for correcting a single error with permutation-invariant codes, and we extend this result to $t\ge1$ errors.

In Sections \ref{sec:preliminaries} and \ref{sec:deletion channel} we collect the necessary definitions and some basic facts about quantum deletions. In particular, in Sec.~\ref{sec:deletion channel} we recall the definition of deletion operators \cite{ouyangEquivalence} and prove some of their properties. Sec.~\ref{sec:KL} contains a detailed form of the error correcting conditions for permutation-invariant codes. Our main result (the new code family) is presented in Sec.~\ref{sec:new family} (checking the error correcting conditions turns out to be technically involved, and the proof is moved to Appendix~\ref{sec:Proofs}).
In Sec.~\ref{sec:spontaneous} we find the conditions on the code parameters for the codes to correct a given number $t\ge 1$ of spontaneous decay errors. Finally, in Sec.~\ref{sec:PR conditions} we present a generalization of the 
Pollatsek--Ruskai conditions for error correction with their permutation-invariant codes, and show that it potentially leads to new examples of such codes.


\section{Preliminaries}\label{sec:preliminaries}
Throughout this paper, we use the following notation. Let $ \ket{\Psi} = \ket{\psi_1\psi_2,\ldots,\psi_n} = 
\ket{\psi_1} \otimes \ket{\psi_2} \otimes \ldots \otimes \ket{\psi_n} \subset \complex^{2\otimes n} $ be a pure state, where $\complex^{2\otimes n}$ is a shorthand for $(\complex^{2})^{\otimes n}$, and we assume that $\braket{\psi_i}{\psi_i} = 1$ for all
$i=1,2,\dots,n$. A general quantum state is identified with its density matrix, i.e.,  
a positive semidefinite Hermitian matrix of trace 1. The density matrix of a pure state is simply
$\rho = \ket{\psi}\bra{\psi} $. For a collection of pure states $\ket{\psi_1}, \ket{\psi_2},\ldots, \ket{\psi_n}$ such that $\Pr(\ket{\psi_i})=p_i$ for all $i$ and $\sum_ip_i=1,$ the density matrix is defined as $\rho = \sum_i 
p_i\ket{\psi_i}\bra{\psi_i} $. Denote by $S( \complex^{2\otimes n}) $ the set of all density matrices of order $ 
2^n $. 

\begin{definition}\label{DefPartialTrace} Consider an $n\times n$ matrix  
$A = \sum_{\bfx,\bfy \in \{ 0,1 \}^n} a_{\bfx,\bfy}\ket{\bfx}\bra{\bfy}$, 
where $ a_{\bfx,\bfy} \in \complex $. For an integer $i \in \{1,2,\dots,n\} $, the {\em partial trace} of $A$ is a mapping
\begin{align*}
    \Tr_i : S( \complex^{2\otimes n} ) &\xrightarrow{} S( \complex^{2\otimes (n-1)} ) \\
    A &\mapsto \sum_{\bfx, \bfy \in\{0,1\}^n} a_{\bfx, \bfy} \Tr(\ket{x_i}\bra{y_i})\ket{\bfx_{\sim i}}\bra{\bfy_{\sim i}},
\end{align*}
where $\bfx_{\sim i} = \ket{x_1,\ldots ,x_{i-1},x_{i+1},\ldots,x_n } $ and $ \bfy_{\sim i} = \ket{y_1,\ldots ,y_{i-1},y_{i+1},\ldots,y_n } $. 
\end{definition}

\new{
Throughout the paper we use following standard definition of binomial coefficients: for a real $x$ and integer $r$ 
  $$
  \binom xr=\begin{cases}
  \frac{x(x-1)\dots(x-r+1)}{r!} &\text{if } r>0\\
  1&\text{if } r=0\\
  0&\text{otherwise}.
  \end{cases}
 $$
 }

Permutation-invariant quantum states are conveniently described in terms of Dicke states \cite{Dicke1,Dicke2,Dicke3}.
 \begin{definition}\label{DefDickeStates}
A {\em Dicke state} $ \ket{D_w^n}$ is a linear combination of all qubit states of length $n$ of 
``Hamming weight'' $w$, i.e.
    \begin{align*}
            \ket{D_w^n} = \frac{1}{\sqrt{\binom{n}{w}}}\sum_{\substack{\bfx\in\{ 0,1 \}^n \\ |\bfx|=w}}\ket{\bfx}.
    \end{align*}
    Sometimes we also use unnormalized Dicke states given by $\ket{H^n_w}= \sqrt{\binom{n}{w}}\ket{D_w^n}$.
\end{definition}
\new{Note that $\braket{D^n_i}{D^n_j} = \delta_{ij}$, where $ \delta_{ij} $ is the Kronecker delta.}

For spin-$\frac{1}{2}$ particles, a Dicke state $\ket {D_n^w}$ can be viewed as a superposition of the tensor product of states of an $n$-particle system in which $w$ particles are in the spin-up, and $n-w$ are in the spin-down configuration; for instance,
\begin{align*}
    \ket{D_1^3} =\frac{\ket{001}+\ket{010}+\ket{100}}{\sqrt{3}} = \frac{\ket{\downdownarrows\uparrow}+\ket{\downarrow\uparrow\downarrow}+\ket{\uparrow\downdownarrows}}{\sqrt{3}}. 
\end{align*}
A quantum code $\cC$ maps a $2^k$-dimensional Hilbert space into a subspace of the $2^n$-dimensional Hilbert space $\complex^{2\otimes n}$, 
i.e.,  it encodes $k$ logical qubits into $n$ physical qubits. Throughout this paper, we will be dealing with two-dimensional codes and denote their basis codewords by $\ket{\bfc_0}$ and $\ket{\bfc_1}$.

The following definition originates with \cite{ruskai-polatsek}.

\begin{definition}\label{DefPICode}
 A \eczoo[permutation-invariant code]{permutation_invariant} is a pair of basis vectors of the form
    \begin{align} \label{eq:codewords}
        \ket{\bfc_0} = \sum_{j=0}^n\alpha_j\ket{D^n_j} \quad \text{and} \quad \ket{\bfc_1} = \sum_{j=0}^n\beta_j\ket{D^n_j},
    \end{align}
where $ \alpha_j,\beta_j\in\complex, j=0,1,\ldots, n$ and $ \sum_j\Bar{\alpha_j}\beta_j = 0.$
\end{definition}
\subsection{Kraus Operators and the Knill--Laflamme conditions}
A {\em quantum channel} $\cA$ is a linear operator acting on density matrices such that it admits the Kraus decomposition
\begin{align}\label{eq:Kraus}
    \mathcal{A}(\rho)=\sum_{\mathbf{A} \in \mathfrak{K}_{\mathcal{A}}} \mathbf{A} \rho \mathbf{A}^{\dagger},
\end{align}
where $ \sum_{\krA}\bfA \bfA^\dagger = {\bfI} $ and $ \krA $ is the Kraus set of the channel. 
Elements of this set are called \textit{Kraus operators}.
\begin{example}(Depolarizing Channel). Let $ \bfX, \bfY, \bfZ$ (bit flip, phase flip, combined flip)
be the set of Pauli errors that act on a state $\rho \in S(\complex^2)$ with probability $p/3$ each. 
Defining $ \krA = \{ ({\bfI}\sqrt{1-p}), (\bfX\sqrt{p/3}),$ $(\bfY\sqrt{p/3}), (\bfZ\sqrt{p/3}) \} $ to be the
Kraus set of the depolarizing channel, we can write its action on $\rho$ in the form
\begin{align*}
    \cA(\rho) = (1-p)\rho + \frac{p}{3}\left(\bfX\rho \bfX + \bfY\rho \bfY + \bfZ\rho \bfZ\right).
\end{align*}
\end{example}
Observe that $ \sum_{\bfA\in\krA}\bfA\bfA^\dagger={\bfI} $. 

The necessary and sufficient conditions for the quantum error correction were formulated by Knill and Laflamme \cite{knill}.
\begin{theorem}[Knill--Laflamme conditions]\label{factKnill} Let $\cC$ be a quantum code with an orthonormal basis
$ \ket{\bfc_0}, \ket{\bfc_1},\ldots, \ket{\bfc_{k-1}} $, and let $ \cA $ be a quantum channel with Kraus operators $ \bfA_i $. There exists a
quantum recovery operator $\cR$ such that $\cR(\cA(\rho))=\rho $ for every density matrix supported on $\cC$ if and only if 
for every $a,b$,
    \begin{align}
        &\bra{\bfc_i}\bfA_a^\dagger \bfA_b \ket{\bfc_j} = 0 \quad \text{for all $ i\neq j $ }, \label{eq:KL1}\\
        &\bra{\bfc_i}\bfA_a^\dagger \bfA_b \ket{\bfc_i} = g_{ab} \quad \text{for all $ i = 0,1,\ldots,k-1 $},
        \label{eq:KL2}
    \end{align}
for some constants $ g_{ab} \in \complex $.
\end{theorem}


\section{Quantum Deletion Channel}\label{sec:deletion channel}
In the classical coding theory, a $t$-deletion error is defined as a map from an $n$-bit string $x$ 
to the set of its subsequences of length $n-t$. 
Following \cite{ouyangEquivalence}, in this section, we define deletions and a deletion channel for quantum codes.   We begin with the following definition.
\begin{definition}\label{def3.1}
($ t $-Deletion channel) Let $t\in\{1,2,\dots,n\}$ and let $\rho\in S(\complex^{2\otimes n})$ be a quantum state. For a set $ I=\{ i_1,i_2\,\ldots,i_t \} \subset \{1,2,\ldots,n\} $ with $ i_1<i_2<\ldots<i_t$, define a map $D_I : S( \complex^{2\otimes n})  \to 
S(\complex^{2\otimes(n-t)}) $ as
\begin{align*}
    D_I^n(\rho):=\Tr_{i_1} \circ \cdots \circ \Tr_{i_t}(\rho).
\end{align*}
The action of $D_I$ deletes the qubits in locations contained in $I$, and is called a $t$-deletion error.
A $t$-{\em deletion channel} $\operatorname{Del}_t^n $ is a convex combination of all $t$-deletion errors, where $t$ is a fixed integer with $ |I|=t $, i.e., 
\begin{align}\label{eq:deletion-channel}
    \operatorname{Del}_t^n(\rho) = \sum_{I : |I|=t}p(I)D_I^n(\rho),
\end{align}
where $ p(I) $ is a probability distribution.
\end{definition}

Let $n\ge t\ge 1$ be integers. Define the set $ I = \{ i_1,i_2,\ldots, i_t \} 
\subset \{1,2,\ldots n\}  $ with $ i_1 < i_2 < \ldots < i_t $. Let $ 
\ket{\bfc}=\ket{c_1c_2\ldots c_t} $ be a pure quantum state with
$ \bfc \in \{ 0,1 \}^t $. Define the operator $ \bfA_{I,\bra{c}}^n = A_1 \otimes A_2 \otimes 
\ldots \otimes A_{n} $ \cite{ouyangEquivalence}, where
\begin{align*}
    A_j = \begin{cases}
    \bra{c_i} \quad j=e_i\in E, \\
    {\bfI}_2 \quad\quad j\notin E.
    \end{cases}
\end{align*}
Here $ {\bfI}_2 $ is the $ 2\times 2 $ identity matrix.
\begin{lemma}[\cite{ouyangEquivalence}, Lemma 3.1] \label{lemma3.1}
 Let $ t\geq 1 $, $ n \geq t $ be integers. Define the set $ I = \{i_1,i_2,\ldots, i_t \} \subset \{1,2,\ldots n\}  $ with $ i_1 < i_2 < \ldots < i_t $. Let $ \ket{\bfc}=\ket{c_1c_2\ldots c_t} $ be a pure quantum state with $ (c_1c_2\ldots c_t) \in \{ 0,1 \}^t $. Then,
\begin{align*}
   \bfA_{I,\bra{\bfc}}^n = \bfA_{i_1,\bra{c_1}}^{n-t+1}\bfA_{i_2,\bra{c_2}}^{n-t+2}\ldots\bfA_{i_{t-1},\bra{c_{t-1}}}^{n-1}\bfA_{i_t,\bra{c_t}}^n .
\end{align*}
\end{lemma}
Lemma \ref{lemma3.1} can be easily proved by direct calculations. 
We are now in a position to describe the Kraus decomposition of the deletion channel. 
First, we cite another auxiliary result.
\begin{lemma}[\cite{ouyangEquivalence}, Lemma 4.2]\label{lemma3.2}
Let $ \rho \in S\left( \complex^{2\otimes n} \right) $ be a quantum state. The output state after deleting the qubits on the positions labeled by the set $ I\subset \{1,2,\ldots,n\} $ can be expressed as
    \begin{align*}
        D_I^n(\rho) = \sum_{\bfc \in \{ 0,1 \}^t} \bfA_{I,\bra{\bfc}}^n \rho \bfA_{I,\bra{\bfc}}^{n^{\dagger}}.
    \end{align*}
\end{lemma}
Lemma ~\ref{lemma3.2} together with Definition \ref{def3.1} implies that the
Kraus decomposition \eqref{eq:Kraus} of the quantum $t$-deletion channel is given by
\begin{align}\label{eq:A}
    \operatorname{Del}_t^n(\rho)=\sum_{I, \bfc} p(I) \bfA_{I,\langle\bfc|}^{n} \rho \bfA_{I,\langle\bfc|}^{n}{ }^{\dagger},
\end{align}
where $ p(E) $ is a probability distribution.

It will be convenient to distinguish between two types of deletions.
\begin{definition} \label{DefDelErrorOperator}
    (Deletion operators) A {\em $0$-type deletion} applied to the $i$-th qubit in an $n$-qubit system 
is the operator $\bfF^{(n)}_i:=\bfA_{i,\bra{0}}^n $. Likewise,  a {\em $1$-type deletion} is the operator $\bfG^{(n)}_i:= \bfA_{i,\bra{1}}^n.$ In other words, given $ \bfx\in\{0,1\}^n $
the action of these operators on the state $\ket \bfx$ is 
    \begin{align*}
        \bfF^{(n)}_i\ket{\bfx} = \braket{0}{x_i}\ket{\bfx_{\sim i}} \quad \text{and} \quad \bfG^{(n)}_i\ket{\bfx}=\braket{1}{x_i}\ket{\bfx_{\sim i}},
    \end{align*}
    where $ \ket{\bfx_{\sim i}} = \ket{x_1\ldots x_{i-1}x_{i+1}\ldots x_n} $.
\end{definition}
At first glance, it is not clear how the channel representation \eqref{eq:A} fits with the definition in \eqref{eq:deletion-channel}. In the next example we illustrate their equivalence for the case of 2-qubit systems.
\new{\begin{example}
    (Single deletion channel) For a $2$-qubit system,
Lemma \ref{lemma3.2} together with Definitions \ref{def3.1} and \ref{DefDelErrorOperator}
imply the following form of the single-deletion channel:
    \begin{align}\label{eq:KrausSingleDel}
        \operatorname{Del}_1^2(\rho) = p_1\left( \bfF_1^{(2)}\rho\bfF_1^{(2)\dagger}+\bfG_1^{(2)}\rho\bfG_1^{(2)\dagger}\right) +  p_2\left( \bfF_2^{(2)}\rho\bfF_2^{(2)\dagger}+\bfG_2^{(2)}\rho\bfG_2^{(2)\dagger}\right),
    \end{align}
    where $ p_1+p_2=1 $. Our point is that each of the two brackets represents a partial trace, see \eqref{eq:ptrace} below.
        On account of \eqref{eq:A},  \eqref{eq:Kraus}, and Definition \ref{DefDelErrorOperator}, Kraus operators for the deletion channel have the following explicit form:
    \begin{align}
        \bfA_1=\sqrt{p_1}\bfF_1^{(2)} =\sqrt{p_1}\begin{bmatrix}
            1&0&0&0\\
            0&1&0&0
        \end{bmatrix}, \quad &\bfA_2=\sqrt{p_2}\bfF_2^{(2)} =\sqrt{p_2}\begin{bmatrix}
            1&0&0&0\\
            0&0&1&0
        \end{bmatrix}, \nonumber\\
        \bfA_3=\sqrt{p_1}\bfG_1^{(2)} =\sqrt{p_1}\begin{bmatrix}
            0&0&1&0\\
            0&0&0&1
        \end{bmatrix}, \quad &\bfA_4=\sqrt{p_2}\bfG_2^{(2)} =\sqrt{p_2}\begin{bmatrix}
            0&1&0&0\\
            0&0&0&1 
        \end{bmatrix}.\label{eq:KrausMatrixSingle}
    \end{align}
    It is easy to verify
    \begin{align*}
        \bfA_1^\dagger\bfA_1+\bfA_2^\dagger\bfA_2+\bfA_3^\dagger\bfA_3+\bfA_4^\dagger\bfA_4 = \bfI_4,
    \end{align*}
where $ \bfI_4 $ is the identity matrix of order 4. Consider the action of the channel on the 
$\ket{\psi}=\frac{1}{\sqrt{2}}( \ket{00} + \ket{01} ) $, whose density matrix is
    \begin{align}\label{eq:DensityMatSingle}
        \rho = \ket{\psi}\bra{\psi} = \begin{bmatrix}
            1/2&1/2&0&0\\
            1/2&1/2&0&0\\
            0&0&0&0\\
            0&0&0&0
        \end{bmatrix}.
    \end{align}
    Inserting \eqref{eq:KrausMatrixSingle} and \eqref{eq:DensityMatSingle} into \eqref{eq:KrausSingleDel}, we obtain
    \begin{align}\label{eq:ptrace}
        \operatorname{Del}_1^2(\rho) &= p_1\begin{bmatrix}
            1/2&1/2\\
            1/2&1/2
        \end{bmatrix} + p_2\begin{bmatrix}
            1&0\\\
            0&0
        \end{bmatrix}=p_1\Tr_1(\rho)+p_2\Tr_2(\rho),
    \end{align}
    meaning that the deletion operation removes on the first qubit with probability $ p_1 $ and the second qubit with probability $ p_2 $.
\end{example}
}

In the next lemma, we present explicit action of the deletion operations on Dicke states.
\new{\begin{lemma}\label{Lemma:DeletionOperatorAction}
    Let $\ket{D_w^n}$ be a Dicke state. Then for all $ i\in\{1,2,\ldots,n\}$,
    \begin{align*}
        \bfF^{(n)}_i\ket{D_w^n} = \sqrt{\frac{\binom{n-1}{w}}{\binom{n}{w}}}\ket{D_w^{n-1}} ,\\
        \bfG^{(n)}_i\ket{D_w^n} = \sqrt{\frac{\binom{n-1}{w-1}}{\binom{n}{w}}}\ket{D_{w-1}^{n-1}} .
    \end{align*}
\end{lemma}
\begin{proof}
    For any $ i \in \{1,2,\ldots, n\} $, acting by $\bfF^{(n)}_i$ on the state $\ket{H^n_w}=\sum_{\bfx: |\bfx|=w}\ket{\bfx}$ annihilates the terms $\ket \bfx$ with $x_i=1$ and deletes one zero from the states $\ket\bfx$ with $x_i=0$. Thus, the only retained states are those with $x_i=0$, and 
  $$
  \bfF^{(n)}_i\ket{H^n_w}=\ket{H^{n-1}_w}.
  $$
Likewise,
  $$
  \bfG^{(n)}_i\ket{H^n_w}=\ket{H^{n-1}_{w-1}}. 
  $$    
\end{proof}}

By the nature of permutation-invariant states, the statements we make below in the paper do not depend on the location of the deleted qubit, and we write $0$-type and $1$-type deletions simply as $\bfF,\bfG$, omitting the subscripts and superscripts from the notation. 

Let us write out explicitly the action of powers of the operators $\bf F$ and $\bfG$ on Dicke states.
\new{\begin{lemma}\label{LemmaOperatorAction}
Let $\ket{D_w^n}$ be a Dicke state and let $a\in \{1,\dots,n\}.$ Then for any $ k \in \{1,2,\ldots, a\} $ and $i_k\in\{1,2,\ldots,n-k+1\}$
    \begin{align*}
        &(\bfF)^a\ket{D^n_w} = \bfF^{(n-a+1)}_{i_a}\ldots\bfF^{(n-1)}_{i_2}\bfF^{(n)}_{i_1}\ket{D^n_w} = \sqrt{\frac{\binom{n-a}{w}}{\binom{n}{w}}}\;\ket{D^{n-a}_w},\\
        &(\bfG)^a\ket{D^n_w} = \bfG^{(n-a+1)}_{i_a}\ldots\bfG^{(n-1)}_{i_2}\bfG^{(n)}_{i_1}\ket{D^n_w} = \sqrt{\frac{\binom{n-a}{w-a}}{\binom{n}{w}}}\;\ket{D^{n-a}_{w-a}}.
    \end{align*}
\remove{where by definition $\binom{n}{r}=0$ if $n<r$ or $r<0$.}
\end{lemma}
\begin{proof} 
    By a direct calculation using Lemma \ref{Lemma:DeletionOperatorAction},
    \begin{align*}
        (\bfF)^a\ket{D^n_w}=
\sqrt{  \frac{\binom{n-a}{w}\dots\binom{n-2}{w}\binom{n-1} w}
     {\binom{n-a+1}{w}\dots \binom{n-1}w\binom nw}}\ket{D^{n-a}_w}=\sqrt{\frac{\binom{n-a}{w}}{\binom{n}{w}}}\ket{D^{n-a}_w}.    
    \end{align*}
 Similarly,
    \begin{align*}
        (\bfG)^a\ket{D^n_w}&=\sqrt{ \frac{\binom{n-a}{w-a}\dots\binom{n-2}{w-2}\binom{n-1}{w-1}}
        {\binom{n-a+1}{w-a}\dots\binom{n-1}{w-1}\binom nw}}\ket{D^{n-a}_{w-a}}
=\sqrt{\frac{\binom{n-a}{w-a}}{\binom{n}{w}}}\ket{D^{n-a}_{w-a}}. \qedhere
    \end{align*}
\end{proof}}

\new{The Kraus operators for the deletion channel will be written as combinations of powers of $F$ and $G$. It helps that the actions of these operators on any permutation-invariant state
commute.
\begin{lemma}\label{Lemma:Commutativity}
    Let $\ket{D_w^n}$ be a Dicke state. For any $ i_1,j_1\in \{1,2,\ldots,n\} $ and $ i_2,j_2\in\{1,2,\ldots,n-1\} $,
    \begin{align*}
        \bfG^{(n-1)}_{i_2}\bfF^{(n)}_{i_1}\ket{D_w^n}=\bfF^{(n-1)}_{j_2}\bfG^{(n)}_{j_1}\ket{D_w^n}.
    \end{align*}
\end{lemma}
\begin{proof}
    By a direct calculation using Lemma \ref{Lemma:DeletionOperatorAction},
    \begin{align*}
         \bfG^{(n-1)}_{i_2}\bfF^{(n)}_{i_1}\ket{D_w^n} =  \sqrt{\frac{\binom{n-1}{w}}{\binom{n}{w}}}\bfG^{(n-1)}_{i_2}\ket{D_w^{n-1}}=\sqrt{\frac{\binom{n-2}{w-1}}{\binom{n}{w}}}\ket{D_{w-1}^{n-2}}.
    \end{align*}
    Similarly,
     \begin{align*}
         \bfF^{(n-1)}_{j_2}\bfG^{(n)}_{j_1}\ket{D_w^n}& =  \sqrt{\frac{\binom{n-1}{w-1}}{\binom{n}{w}}}\bfF^{(n-1)}_{j_2}\ket{D_{w-1}^{n-1}}=\sqrt{\frac{\binom{n-2}{w-1}}{\binom{n}{w}}}\ket{D_{w-1}^{n-2}}. \qedhere
   \end{align*}
\end{proof}
}

\new{Lemmas \ref{lemma3.1}, \ref{lemma3.2}, \ref{LemmaOperatorAction}, and \ref{Lemma:Commutativity}} imply that
the Kraus set of the $ t$-deletion channel for a permutation-invariant code has the form
$\{ \bfG^a\bfF^{t-a} : a\in\{ 0,1,\ldots, t \} \}$. Throughout the paper, we will write this set as
\begin{align} \label{eq:Kraus deletion}
    \varepsilon_t = \{ \E_0, \E_1 , \ldots, \E_t \},
\end{align}
where $ \E_a = \bfG^a\bfF^{t-a} $. The following lemma describes the action of the error operator $ \E_a\in\varepsilon_t $ on permutation-invariant states. 
\begin{lemma}\label{LemmaKrausOperators}
Let $ \E_a $ be an element of the Kraus set of the $ t$-deletion channel $\varepsilon_t$ for a permutation-invariant code. Then, for any permutation-invariant state $\ket{D^n_w}$, 
    \begin{align*}
        \E_a\ket{D^n_w} = \sqrt{\frac{\binom{n-t}{w-a}}{\binom{n}{w}}}\ket{D^{n-t}_{w-a}}.
    \end{align*}
\end{lemma}
\begin{proof}
    By Lemma \ref{LemmaOperatorAction},
    \begin{align*}
        \E_a\ket{D^n_w} = \bfG^a\bfF^{t-a}\ket{D^n_w} &= \sqrt{\frac{\binom{n-t+a}{w}}{\binom{n}{w}}}\bfG^a\ket{D^{n-t+a}_w}\\
        &=\sqrt{\frac{\binom{n-t+a}{w}}{\binom{n}{w}}}\sqrt{\frac{\binom{n-t+a-a}{w-a}}{\binom{n-t+a}{w}}}\ket{D^{n-t+a-a}_{w-a}}\\ 
        &=  \sqrt{\frac{\binom{n-t}{w-a}}{\binom{n}{w}}}\ket{D^{n-t}_{w-a}}.  \qedhere
    \end{align*}
\end{proof}

We make an important observation: Upon applying a deletion error to a permutation-invariant state,
we obtain a permutation-invariant state (on fewer qubits). Clearly, this does not hold for Pauli errors: for instance,
  $$ 
  \bfX_1\ket{D^3_1} = 
\frac{\ket{101}+\ket{110}+\ket{000}}{\sqrt{3}},
  $$
which is not permutation-invariant. Moreover, generally $X_i\ket{D_w^n}\ne X_j\ket{D_w^n}$ if $i\ne j$,
so the Kraus set for Pauli errors is much larger than for deletions, complicating the
analysis. A workaround proposed in \cite{ruskai-polatsek} suggests averaging
Pauli errors, but general constructions look difficult. At the same time, invariance
with respect to permutations plays the defining role for deletions, and it is also the
main property supporting the code construction we propose. We note that given a
deletion-correcting permutation-invariant code, we can argue about its distance and
make claims about its properties with respect to correcting Pauli errors. 
Indeed, the following proposition is true.
 \begin{proposition}\label{factErasureDeletionEquivalence}
A permutation-invariant code that corrects $2t$ deletions, also corrects all combinations of
$t$ Pauli errors.
    \end{proposition}
\begin{proof} 
For a permutation-invariant state, deleting any $2t$ qubits is equivalent to deleting the first $2t$ qubits
in an $n$-qubit state, so deletions are equivalent to erasures. Of course, a code that
corrects $2t$ erasures has the quantum distance of at least 2$t+1$ (see, e.g., \cite{rains}). Thus,
correcting deletions is tied to the code distance, and distance $d=2t+1$ is a sufficient condition for correcting $t$ qubit errors.
\end{proof}


\section{Error correction conditions for permutation-invariant codes}\label{sec:KL}
\remove{
General conditions for a code to correct $t$ deletions were phrased in \cite{hagiwaraDeletion}, but using
them to construct codes is not immediate. In this section we focus on permutation-invariant codes and derive simple
conditions for such a code to correct deletions (by the above, this will also imply that they correct qubit errors).
}
Sufficient conditions for any code to correct $t$ deletions were previously derived in \cite{hagiwaraDeletion}. 
In this section, we focus on permutation-invariant codes and derive the necessary and sufficient conditions for such a code to correct deletions by showing the
equivalence between them and the Knill--Laflamme conditions for the $2t$-deletion channel. By Proposition \ref{factErasureDeletionEquivalence}, this also implies that they correct $t$ qubit errors.
\begin{theorem}\label{theoremConditions}
    Let $ \cC $ be a permutation-invariant quantum error correction code as given in Definition \ref{DefPICode}, and suppose that the coefficients $\alpha_j$ and $\beta_j$, $j=1,\dots,n$ in the codewords \eqref{eq:codewords} are real. Then the code $ \cC$ corrects all $ t$-qubit errors if and only if its coefficient vectors $ \underline{\alpha}= (\alpha_0,\alpha_1,\ldots,\alpha_n) $ and $ \underline{\beta} = (\beta_0,\beta_1,\ldots,\beta_n) $ satisfy the following conditions:
          \begin{align*}
        &\text{\rm(C1)}\quad \sum_{j=0}^n\alpha_j\beta_j = 0;\\
        &\text{\rm(C2)}\quad \sum_{j=0}^n\alpha_j^2 = \sum_{j=0}^n\beta_j^2 = 1;\\
        &\text{\rm(C3) \;For all $ 0\leq  a,b  \leq 2t$},\\
        & \hspace*{1in}\sum_{j=0}^n\frac{\binom{n-2t}{j}}{\sqrt{\binom{n}{j+a} \binom{n}{j+b}}} \alpha_{j+a}\beta_{j+b}=0;\\
        &\text{\rm(C4) \;For all $ 0\leq  a,b  \leq 2t$},\\
        & \hspace*{1in}\sum_{j=0}^n\frac{\binom{n-2t}{j}}{\sqrt{\binom{n}{j+a} \binom{n}{j+b}}}\left(\alpha_{j+a}\alpha_{j+b}-\beta_{j+a}\beta_{j+b}\right)=0,
    \end{align*}
   \new{ In {\rm(C3), (C4)} we assume by definition that $\frac{\alpha_s}{\sqrt{\binom{n}{s}}}=\frac{\beta_s}{\sqrt{\binom{n}{s}}}=0 $ if $ s>n $.}
\end{theorem}
\begin{proof}
Since the Dicke states are orthonormal by construction, conditions (C1), (C2)
are required for the codewords $\ket{\bfc_0},\ket{\bfc_1}$ to form
orthonormal states. We will argue that conditions (C3) and (C4) are equivalent to the Knill--Laflamme conditions for the $2t$-deletion channel. 
By Proposition~\ref{factErasureDeletionEquivalence} this suffices to prove the theorem. 

Recall the form of the Kraus set of the $ 2t$-deletion channel \eqref{eq:Kraus deletion}.
By \eqref{eq:codewords} and Lemma \ref{LemmaKrausOperators}, we have
    \begin{align*}
        &\E_a\ket{\bfc_0} = \sum_{j=0}^{n}\alpha_j\sqrt{\frac{\binom{n-2t}{j-a}}{\binom{n}{j}}}\ket{D^{n-2t}_{j-a}},\\
        &\E_a\ket{\bfc_1} = \sum_{j=0}^{n}\beta_j\sqrt{\frac{\binom{n-2t}{j-a}}{\binom{n}{j}}}\ket{D^{n-2t}_{j-a}}.
    \end{align*}
To show the equivalence of (C3) and \eqref{eq:KL1}, we compute
    \begin{align*}
\bra{\bfc_0}\E_a^\dagger \E_b \ket{\bfc_1}&=\sum_{j=0}^n\sum_{j'=0}^n\alpha_j\beta_{j'}\sqrt{\frac{\binom{n-2t}{j-a}\binom{n-2t}{j'-b}}{\binom{n}{j}\binom{n}{j'} }} \braket{D^{n-2t}_{j-a}}{D^{n-2t}_{j'-b}}\\
&=\sum_{j=0}^n\sum_{j'=0}^n\alpha_j\beta_{j'}\sqrt{\frac{\binom{n-2t}{j-a}\binom{n-2t}{j'-b}}{\binom{n}{j}\binom{n}{j'} }}\delta_{j-a,j'-b}\\
&= \sum_{j=0}^n\frac{\binom{n-2t}{j}}{\sqrt{\binom{n}{j+a} \binom{n}{j+b}}} \alpha_{j+a}\beta_{j+b}
    \end{align*}
for all $ 0\leq a,b\leq 2t $. 
Turning to condition (C4), we find
    \begin{align*}
\bra{\bfc_0}\E_a^\dagger \E_b \ket{\bfc_0}-\bra{\bfc_1}\E_a^\dagger \E_b \ket{\bfc_1} 
&=\sum_{j=0}^n\sum_{j'=0}^n\sqrt{\frac{\binom{n-2t}{j-a}\binom{n-2t}{j'-b}}{\binom{n}{j}\binom{n}{j'} 
}}\left( \alpha_j\alpha_{j'} - \beta_j\beta_{j'} \right) \braket{D^{n-2t}_{j-a}}{D^{n-2t}_{j'-b}}\\
&=\sum_{j=0}^n\sum_{j'=0}^n\sqrt{\frac{\binom{n-2t}{j-a}\binom{n-2t}{j'-b}}{\binom{n}{j}\binom{n}{j'} 
}}\left( \alpha_j\alpha_{j'} - \beta_j\beta_{j'} \right) \delta_{j-a,j'-b}\\
&=\sum_{j=0}^n\frac{\binom{n-2t}{j}}{\sqrt{\binom{n}{j+a} \binom{n}{j+b}}}\left(\alpha_{j+a}\alpha_{j+b}-
\beta_{j+a}\beta_{j+b}\right)
    \end{align*}
for all $ 0\leq a,b\leq 2t$, and thus (C4) is equivalent to \eqref{eq:KL2}. 
Together with Proposition~\ref{factErasureDeletionEquivalence} this proves the theorem.
\end{proof}
\new{\begin{example}\label{example : gnu codes}
 Ouyang's permutation-invariant codes \cite{ouyangPI} with parameters $ (g,m,u) $ can be defined via the logical computational basis 
\begin{align*}
    \ket{\bfc_0} = \sum_{\substack{\text{$l$ even}\\0\leq l \leq m}}\sqrt{\frac{\binom{m}{l}}{2^{m-1}}}\ket{D^n_{gl}},\quad
    \ket{\bfc_1} = \sum_{\substack{\text{$l$ odd}\\0\leq l \leq m}}\sqrt{\frac{\binom{m}{l}}{2^{m-1}}}\ket{D^n_{gl}},
\end{align*}
where $ n=gmu $ is the code length. Consider a $ (2t+1,2t+1,1)$ code from this family. Its coefficient vectors trivially satisfy conditions (C1)-(C3) because of the choice of the gap parameter $g=2t+1$, and  condition (C4) turns into 
\begin{align*}
    \sum_{l=0}^m(-1)^l\binom{m}{l}\frac{\binom{n-2t}{gl-a}}{\binom{n}{gl}},
\end{align*}
which is zero for all $ 0\leq a \leq 2t $ (see Lemmas 1 and 2 in \cite{ouyangPI}). Coupled with Theorem 
\ref{theoremConditions}, this shows that this code corrects $ t $ qubit errors, 
recovering one of the results in \cite{ouyangPI}. This example will prove useful in Prop.~\ref{prop:gm1} below, where we relate our code construction to Ouyang's codes.
\end{example}}


\section{A new family of permutation-invariant Codes}\label{sec:new family}
In this section, we present a new family of permutation-invariant codes, defined by the parameters $g,m,\delta$, and $\epsilon$. Here, $g$ is
the gap parameter, $m$ is the occupancy number, $\delta$ is a parameter to adjust the code length, and the parameter $  \epsilon $ determines the sign of the coefficients. The code we construct encodes one logical qubit into $n=2gm+\delta+1$ physical qubits. 
The following combinatorial identities, proved in Appendix~\ref{sec:Proofs}, will play a role in the construction.
\begin{lemma}\label{lemma:C1}
    Let $n, g, m, a ,r$ be integers such that $g>0$ and $0\leq a \leq r \leq 2m < n/g$. Then 
    \begin{align}
        \sum_{l=0}^{m}(-1)^l\frac{\binom{m}{l}}{\binom{n/g-l}{m+1}}\left(\frac{\binom{n-r}{gl-a}}{\binom{n}{gl}}-\frac{\binom{n-r}{gl-r+a}}{\binom{n}{gl}}\right)=0. \label{eq:E1}
    \end{align}
\end{lemma}

\begin{lemma}\label{lemma:C2}
     For any real $x$ and integer $m$ such that $x>m>0,$
     \begin{equation}\label{eq:Z}
         \sum_{l=0}^m\frac{\binom{m}{l}}{\binom{2x-l}{m+1} } = \binom{x}{m}^{-1}\frac{(m+1)}{2(x-m)}.
     \end{equation}
\end{lemma}

\begin{construction}\label{constructionGmdelta}
Let $g,m,\delta $ be nonnegative integers, and let $\epsilon\in\{-1,+1\}$. Define a permutation-invariant code $\cQ_{g,m,\delta,\epsilon}$ via its logical computational basis
    \begin{align*}
        &\ket{\bfc_0} =\sum_{\substack{\text{$l$ {\rm even}}\\0\leq l \leq m}} \gamma b_l\ket{D^n_{gl}} + 
        \sum_{\substack{\text{$l$ {\rm odd}}\\0\leq l \leq m}} \gamma b_l\ket{D^n_{n-gl}},\\
        &\ket{\bfc_1} = \sum_{\substack{\text{$l$ {\rm odd}}\\0\leq l \leq m}} \gamma b_l\ket{D^n_{gl}} +\epsilon
        \sum_{\substack{\text{$l$ {\rm even}}\\0\leq l \leq m}} \gamma b_l\ket{D^n_{n-gl}},
    \end{align*}
    where $ n=2gm+\delta+1, $
    $
    b_l=\sqrt{{\binom{m}{l}}/{\binom{n/g-l}{m+1} }},
    $
and $ \gamma =  \sqrt{\binom{n/(2g)}{m} \frac{n-2gm}{g(m+1)} } $ is the normalizing factor. 
\end{construction}
The next theorem establishes the error correction properties of the code $\cQ_{g,m,\delta,\epsilon}$.
\new{\begin{theorem}\label{theoremGMdelta}
   Let $ t $ be a nonnegative integer and let $ m\geq t $ and $\delta\geq 2t$. If
    $$
(   g\ge 2t, \epsilon=-1) \text{ or }(g\ge 2t+1,\epsilon=+1),
   $$
then the code $\cQ_{m,l,\delta,\epsilon}$ encodes one qubit into $n = 2gm+\delta+1 $ qubits and corrects any $t$ qubit errors. 
\end{theorem}}
\begin{proof} 
We need to prove that the coefficient vectors $\underline{\alpha},\underline{\beta}$ of the basis states 
satisfy conditions (C1)-(C4). Writing these coefficients for the code  $\cQ_{m,l,\delta,-}$ explicitly, we obtain
    \begin{align*}
        \alpha_j = \sum_{\substack{\text{$l$ even}\\0\leq l \leq m}}f(l)\delta_{j,gl} +  \sum_{\substack{\text{$l$ odd}\\0\leq l \leq m}}f(l)\delta_{j,n-gl}, \\
        \beta_j = \sum_{\substack{\text{$l$ odd}\\0\leq l \leq m}}f(l)\delta_{j,gl} -  \sum_{\substack{\text{$l$ even}\\0\leq l \leq m}}f(l)\delta_{j,n-gl} ,
    \end{align*}
    where $f(l) = \gamma b_l$.  
Throughout the proof, all the sums on $l$ are taken over $l=0,1,\dots, m$.
We start with
        \begin{multline*}
            \alpha_j\beta_j = \sum_{\text{ $l$ even}}\sum_{\text{ $ l^\prime$odd}}f(l)f(l^\prime)\delta_{gl,gl^\prime}\delta_{j,gl}- \sum_{\text{ $ l $ even}}\sum_{\text{ $l^\prime $ even}}f(l)f(l^\prime)\delta_{gl,n-gl^\prime}\delta_{j,gl}\\ +\sum_{ \text{ $ l $ odd}}\sum_{\text{ $l^\prime $ odd}}f(l)f(l^\prime)\delta_{n-gl,gl^\prime }\delta_{j,n-gl}-\sum_{\text{ $l$ odd}}\sum_{\text{ $ l^\prime  $ even}}f(l)f(l^\prime)\delta_{n-gl,n-gl^\prime}\delta_{j,n-gl},
        \end{multline*}
    where we use the fact $ \delta_{jk}\delta_{jl}=\delta_{jl}\delta_{kl}$. The first sum in this expression is
clearly zero since $ gl \neq gl^\prime  $ if $l$ is even and $l^\prime$ is odd, and thus $\delta_{gl , gl^\prime} =0.$. Turning to the second sum, now $ l $ and $ l^\prime $ are even.
It is easy to see that $n=2gm+\delta+1 \neq g(l+l^\prime)$ since $ l+l^\prime \leq 2m,$ and so
$ \delta_{gl,n-gl^\prime}=0$, and the second sum is zero. The third and fourth sums are treated as the
first and second, respectively, and they are easily seen to be zero. This shows that condition (C1) holds. 

To check condition (C2), write
 \begin{multline*}
            \alpha_j^2 = \sum_{\text{$ l$ even}}\sum_{ \text{ $ l^\prime $ even}}f(l)f(l^\prime)\delta_{gl,gl^\prime}\delta_{j,gl}+ \sum_{\text{ $ l $ even}}\sum_{\text{ $ l^\prime $ odd}}f(l)f(l^\prime)\delta_{gl,n-gl^\prime}\delta_{j,gl}\\ +\sum_{\text{ $ l $ odd}}\sum_{\text{ $ l^\prime $ even}}f(l)f(l^\prime)\delta_{n-gl,gl^\prime }\delta_{j,n-gl}+\sum_{\text{ $ l $ odd}}\sum_{\text{ $ l^\prime $ odd}}f(l)f(l^\prime)\delta_{n-gl,n-gl^\prime}\delta_{j,n-gl}.
       \end{multline*}
Notice that the second and third sums are zero since $n=2gm+\delta+1\neq g(l+l^\prime) $ for any $ l,l^\prime \leq m $. The first sum is not 0 only if $ l=l^\prime$, and therefore it can be written as $ \sum_{\text{ $ l $ even}}f(l)^2\delta_{j,gl} $. Similarly, the fourth sum can be expressed as $ \sum_{ \text{ $ l $ odd}}f(l)^2\delta_{j,n-gl} $. Hence,
    \begin{align*}
        \sum_{j=0}^n \alpha_j^2 &= \sum_{j=0}^n\Big( \sum_{\text{$l$ even}}f(l)^2\delta_{j,gl} + \sum_{\text{$l$ odd}}f(l)^2\delta_{j,n-gl} \Big)=\sum_{l=0}^m f(l)^2 \\
        &=\gamma^2\sum_{l=0}^m b_l^2 = 1,
    \end{align*}
where we used Lemma \ref{lemma:C2}. A similar sequence of steps confirms that $\sum_{j=0}^n \beta_j^2=1,$ and thus
condition (C2) is also satisfied.

    For condition $ (C3) $, let us first evaluate the term
    \begin{eqnarray*}
        \begin{multlined}
            \alpha_{j+a}\beta_{j+b} = \sum_{\text{ $ l $ even}}\sum_{\text{ $ l^\prime  $ odd}}f(l)f(l^\prime)\delta_{gl-a,gl^\prime-b}\delta_{j,gl-a}- \sum_{\text{ $ l $ even}}\sum_{\text{ $ l^\prime  $ even}}f(l)f(l^\prime)\delta_{gl-a,n-gl^\prime-b}\delta_{j,gl-a}\\ +\sum_{\text{ $ l $ odd}}\sum_{\text{ $ l^\prime $odd}}f(l)f(l^\prime)\delta_{n-gl-a,gl^\prime-b }\delta_{j,n-gl-a}-\sum_{\text{ $ l $ odd}}\sum_{\text{ $ l^\prime  $even}}f(l)f(l^\prime)\delta_{n-gl-a,n-gl^\prime-b}\delta_{j,n-gl^\prime-b}.
        \end{multlined}   
    \end{eqnarray*}
First, observe that the second sum is zero. To see this, recall that $ \delta\geq 2t $, $l+l^\prime\leq 2m $, and $ b-
a\leq 2t $. Therefore, $ g(l+l^\prime) + b-a \leq 2gm+2t < 2gm+\delta+1 = n $. Similarly, the third sum is also
zero. For the first sum to be nonzero, it should be that $ g(l-l^\prime)= a-b.$
 If $ g>2t $, then $ |l-l^\prime|\geq 1 $ and $ |a-b| \leq 2t $, so this condition cannot be fulfilled. 
The equality $g(l-l^\prime)= a-b $ is possible only if $g=2t$ and $a-b = \pm 2t$, equivalently $ l^\prime = l \mp 1 $. By the same argument, the fourth sum is not 0 only if $g=2t$ and $ a-b = \pm 2t $, hence 
$ l^\prime = l \pm 1 $. Therefore, the first sum can be written as 
$\sum_{\text{$l$ even}}f(l)f(l\mp1)\delta_{j,gl-
a} $, while the fourth sum can be written as $\sum_{\text{$ l $ even}}f(l)f(l\mp1)\delta_{j,n-
gl-b} $. Hence, we have
    \begin{align*}
        \sum_{j=0}^n\frac{\binom{n-2t}{j}}{\sqrt{\binom{n}{j+a} \binom{n}{j+b}}} \alpha_{j+a}\beta_{j+b} &= \sum_{\text{$l$ even}}f(l)f(l\mp1)\sum_{j=0}^n\frac{\binom{n-2t}{j}}{\sqrt{\binom{n}{j+a} \binom{n}{j+b}}}\left(\delta_{j,gl-a}- \delta_{j,n-gl-b}  \right)\\
        &= \sum_{\text{$l$ even}}f(l)f(l\mp1)\biggl(\frac{\binom{n-2t}{gl-a}}{\sqrt{\binom{n}{gl} \binom{n}{gl+b-a}}}-\frac{\binom{n-2t}{n-gl-b}}{\sqrt{\binom{n}{n-gl-b+a} \binom{n}{n-gl}}}  \biggr)\\
        &= \sum_{\text{$l$ even}}f(l)f(l\mp1)\biggl(\frac{\binom{n-2t}{gl-a}}{\sqrt{\binom{n}{gl} \binom{n}{gl+b-a}}}-\frac{\binom{n-2t}{gl+b-2t}}{\sqrt{\binom{n}{gl+b-a} \binom{n}{gl}}}\biggr).
    \end{align*}
Now, observe that we have two different cases: either $a=0$ and $b=2t$, or $ a=2t $ and $ b=0 $. For both of
them, the difference inside the parentheses is zero, meaning that we have proved condition (C3).

Finally, consider condition $(C4).$ Let us start with evaluating the term
        \begin{multline*}
            \alpha_{j+a}\alpha_{j+b} =  \sum_{\text{ $l $even}}\sum_{\text{$l^\prime$even}}f(l)f(l^\prime)\delta_{gl-a,gl^\prime-b}\delta_{j,gl-a} + \sum_{\text{$l$even}}\sum_{\text{$l^\prime $ odd}}f(l)f(l^\prime)\delta_{gl-a,n-gl^\prime-b}\delta_{j,gl-a}\\ + \sum_{\text{ $ l $ odd}}\sum_{\text{ $ l^\prime  $ odd}}f(l)f(l^\prime)\delta_{n-gl-a,n-gl^\prime-b }\delta_{j,n-gl-a}\\+\sum_{\text{ $ l $ odd}}\sum_{\text{ $ l^\prime  $ even}}f(l)f(l^\prime)\delta_{n-gl-a,gl^\prime-b}\delta_{j,n-gl-a}.
        \end{multline*}
First, recall that $ \delta\geq 2t, $ $ l+l^\prime \leq 2m $ and $ |b-a|\leq 2t $. Therefore, the second sum is zero since $ g(l+l^\prime) +b-a \leq 2gm+2t < 2gm+\delta+1 = n,$ and the fourth sum is zero by a similar argument. 
For the first sum to be nonzero, for both $ l $ and $ l^\prime $ even, the equality $ g(l-l^\prime)= a-b $ should hold. Since $ a-b\leq 2t $ and $ g\geq 2t $, it can hold only when
$ a=b $ and $ l=l^\prime $. As before, what remains of the sum is the diagonal, and it can be written as
$ \sum_{\text{$ l $ even}}f(l)^2\delta_{j,gl-a} $. For the same reasons, the third sum degrades to
$ \sum_{\text{$ l$ odd}}f(l)^2\delta_{j,n-gl-a}  $, which yields
   \begin{align*}
        \alpha_{j+a}\alpha_{j+b} =  \sum_{\text{ $l $ even}}f(l)^2\delta_{j,gl-a}+\sum_{\text{ $ l $ odd}}f(l)^2\delta_{j,n-gl-a}.
    \end{align*}
By following similar steps, we find that
    \begin{align*}
        \beta_{j+a}\beta_{j+b} = \sum_{\text{ $ l $ even}}f(l)^2\delta_{j,n-gl-a}+\sum_{\text{ $ l $ odd}}f(l)^2\delta_{j,gl-a}.
    \end{align*}
In summary, we have
    \begin{align*}
        \alpha_{j+a}\alpha_{j+b}- \beta_{j+a}\beta_{j+b} =  \sum_{l=0}^m(-1)^lf(l)^2\biggl( \delta_{j,gl-a} - \delta_{j,n-gl-a} \biggr).
    \end{align*}
Then, recalling that $a=b$, the expression in condition $(C4)$ has the form
    \begin{align*}
        \sum_{j=0}^n\frac{\binom{n-2t}{j}}{\sqrt{\binom{n}{j+a} \binom{n}{j+b}}}\left(\alpha_{j+a}\alpha_{j+b}-\beta_{j+a}\beta_{j+b}\right)&=
        \sum_{l=0}^m(-1)^lf(l)^2\sum_{j=0}^n\frac{\binom{n-2t}{j}}{\binom{n}{j+a}}\left( \delta_{j,gl-a} - \delta_{j,n-gl-a} \right)\\
        &=\gamma^2\sum_{l=0}^m(-1)^l\frac{\binom{m}{l}}{\binom{n/g-l}{m+1} }\left(\frac{\binom{n-2t}{gl-a}}{\binom{n}{gl}} - \frac{\binom{n-2t}{gl-2t+a}}{\binom{n}{gl}}  \right),
    \end{align*}
which is zero by Lemma~\ref{lemma:C1}.  Following the same sequence of steps, it is possible
to show the error correction property of the code $ \cQ_{m,l,\delta,+} $. The proof is now complete.
\end{proof}

\begin{example}
(PI Code $\cQ_{2,1,2,-}$) Suppose $ g=2 $, $ m=1 $, $ \delta=2 $, and $ \epsilon=-1 $. Then the length of the code is $ n=2gm+\delta+1=7 $ and $ \gamma = \sqrt{\binom{n/(2g)}{m} \frac{n-2gm}{g(m+1)} } = \sqrt{21}/4$. 
The coefficients $b_l$ have the form
   $$ 
   b_0 = \sqrt{{\binom{m}{0}}/{\binom{n/g}{m+1} }} = \sqrt{\frac8{35}}, \quad b_1=\sqrt{{\binom{m}{1}}/{\binom{n/g-1}{m+1} }} = \sqrt{\frac8{15}}.
   $$
   Using Construction $ \ref{constructionGmdelta} $, the code $ \cQ_{2,1,2,-} $ can be defined via its logical codewords
    \begin{align}
    \ket{\bfc_0} = \sqrt{\frac{3}{10}}\ket{D^7_0} +  \sqrt{\frac{7}{10}}\ket{D^7_5} \quad \text{and} \quad \ket{\bfc_1} = \sqrt{\frac{7}{10}}\ket{D^7_2} -  \sqrt{\frac{3}{10}}\ket{D^7_7} \label{eq:Q212}
    \end{align}
Note that it has the same length as the $ 7$-qubit permutation-invariant code of \cite{ruskai-polatsek}, and it can correct a single error owing to Theorem \ref{theoremGMdelta}. 
\end{example}
\begin{example}
(PI Code $\cQ_{4,2,4,-}$) Suppose $ g=4 $, $ m=2 $, $ \delta=4 $, and $ \epsilon=-1 $. Then the length of the code is $ n=2gm+\delta+1=21 $ and $ \gamma = \sqrt{455/512}$. Since $ m=2 $, $ l $ takes values $ 0$ , $ 1 $, and $ 2 $. The coefficients $b_l$ have the form $ b_0 = \sqrt{128/1547} $, $ b_1= 16/\sqrt{663}$, $b_2=\sqrt{128/195}$.
Using Construction \ref{constructionGmdelta}, the code $ \cQ_{4,2,4,-} $ is
\begin{align*}
    &\ket{\bfc_0} = \sqrt{\frac{5}{68}}\ket{D^{21}_0} +  \sqrt{\frac{7}{12}}\ket{D^{21}_8} + \sqrt{\frac{35}{102}}\ket{D^{21}_{17}},\\
    &\ket{\bfc_1} = \sqrt{\frac{35}{102}}\ket{D^{21}_4} -  \sqrt{\frac{7}{12}}\ket{D^{21}_{13}} - \sqrt{\frac{5}{68}}\ket{D^{21}_{21}}    
\end{align*}
This code is shorter than all currently known explicit permutation-invariant codes that correct double errors
\end{example}
Note that the permutation-invariant code $ \cQ_{2t,t,2t,-} $ of length $ (2t+1)^2-2t $ corrects arbitrary $ t $ qubit errors and has the best code parameters among all the previously known permutation-invariant codes with this property. The following proposition describes the relation between our code family and Ouyang's {\em gnu} codes \cite{ouyangPI}.
\new{\begin{proposition}\label{prop:gm1}
For all odd integers $ m>0 $ and for all integers $ g>0 $, the code 
$ \cQ_{g,\frac{m-1}2,g-1,+} $  coincides with Ouyang's {\em gnu} code with parameters $ (g,m,1)$.
\end{proposition}
\begin{proof}
       First notice that the length of the code $ \cQ_{g,\frac{m-1}2,g-1,+} $ is $ n = 2g\frac{m-1}{2}+g-1+1 = gm$, 
matching the length of Ouyang's code with parameters $ (g,m,1) $. Writing the entries of the coefficient vector $ \underline{\alpha} $ for the code $ \cQ_{g,\frac{m-1}2,g-1,+} $, we have
    \begin{align*}
        \alpha_j &= \sum_{\substack{\text{$l$ even}\\0\leq l \leq (m-1)/2}}f(l)\delta_{j,gl} +  \sum_{\substack{\text{$l$ odd}\\0\leq l \leq (m-1)/2}}f(l)\delta_{j,g(m-l)}\\
        &=\sum_{\substack{\text{$l$ even}\\0\leq l \leq (m-1)/2}}f(l)\delta_{j,gl} +  \sum_{\substack{\text{$l^\prime$ even}\\(m+1)/2\leq l^\prime \leq m}}f(l^\prime)\delta_{j,gl^\prime}\\
        &= \sum_{\substack{\text{$l$ even}\\0\leq l \leq m}}f(l)\delta_{j,gl}, 
    \end{align*}
where we made the change of variable $ l^\prime = m-l $ and used the fact that $ m-l $ is even for all $ l $ and $ m $ odd. Computing the function $ f(l)^2 $, we obtain
    \begin{align*}
        f(l)^2 = \gamma^2b_l^2 &=\frac{2}{m+1}\frac{\binom{m/2}{(m-1)/2}\binom{(m-1)/2}{l}}{\binom{m-l}{(m+1)/2}}\\
        &{=}\frac{2}{m+1}\frac{\binom{m/2}{(m-1)/2}}{\binom{m}{(m+1)/2}}\binom{m}{l}\\
        &{=}\frac{1}{2^{m-1}}\binom{m}{l},
    \end{align*}
where the second equality is obtained by rewriting the numerator on the first line, and the third one by writing out
the binomial coefficients on the second line, namely
 $ \binom{m/2}{(m-1)/2}=\frac{m(m-2)(m-4)\dots 1}{(m-1)(m-3)\ldots 2}  $ and  $ \binom{m}{(m+1)/2} = \frac{2^m m!}{(m+1)(m-1)^2(m-3)^2\ldots 2^2} .$   
Therefore, the coefficient vector $ \underline{\alpha} $ of the code $ \cQ_{g,\frac{m-1}2,g-1,+} $ is equal to the coefficient vector $ \underline{\alpha} $ of the code in Example \ref{example : gnu codes} with parameters $ (g,m,1) $. A similar argument establishes that the coefficient vector $ \underline{\beta} $ of two codes are also equal. 
\end{proof}}

\subsection{Deletion Correction Property}
We already know that the code $\cQ_{g,m,\delta,\epsilon}$ corrects deletions. A precise formulation of this claim
is given in the following proposition.
\begin{proposition}
 If $m\geq \ceil{\frac{s}{2}}, \delta\geq s$ and
   $$
  (g\ge s, \epsilon=-1) \text{ or } (g\geq s+1, \epsilon=+1),
  $$
then the code  $\cQ_{g,m,\delta,\epsilon}$ corrects all patterns of $s$ deletions. 
\end{proposition}
This claim follows from the fact that conditions (C3) and (C4) are equivalent to the Knill--Laflamme
conditions for correcting $2t$ deletions when they are phrased for a permutation-invariant code.

\remove{\textit{Proof.} Recall that in the proof of Theorem \ref{theoremConditions}, we proved that conditions $ (C3) $ and $ (C4) $ are equivalent to the Knill--Laflamme conditions for the $ 2t $-deletion channel for a permutation-invariant code. Therefore, it is easy to see that changing the variable $ 2t $ with $ s $ in conditions $ (C3) $ and $ (C4) $ will give us error correction conditions for a permutation-invariant code for the $s$-deletion channel. After that the proof can be easily done by repeating exactly the same steps in the proof of Theorem \ref{theoremGMdelta}. }

For an odd number of deletions, the shortest code $\cQ_{g,m,\delta,\epsilon}$ has length $ (s+1)^2$. This coincides
with the length of Ouyang's {\em gnu} codes, although the code families are different.
For any even number of deletions $>0$, the code $\cQ_{g,m,\delta.\epsilon},$ where $(g,m,\delta,\epsilon)=(s,\ceil{s/2},s,-),$ has length $(s+1)^2-s$, which is shorter than the existing constructions. 

In \cite{hagiwaraSingle}, Nakayama and Hagiwara showed that the smallest length of single quantum deletion-correcting codes is $ 4 $. They also constructed a code that meets this bound with equality.
We note that code $\cQ_{1,1,1,-}$ gives another construction of an optimal code correcting one deletion. Its
logical codewords are
\begin{align*}
    \ket{\bfc_0} &= \sqrt{\frac{1}{3}}\ket{0000} + \sqrt{\frac{1}{6}}\left(\ket{1110} + \ket{1101} + \ket{1011} + \ket{0111} \right),\\
    \ket{\bfc_1} &= \sqrt{\frac{1}{6}}\left(\ket{0001} + \ket{0010} + \ket{0100} + \ket{1000}\right) - \sqrt{\frac{1}{3}}\ket{1111}.
\end{align*}
\subsection{Transversality}
\new {In this section we make some remarks concerning the transversal action of logical gates on the codes that we propose. 
Let $ G+\tau $ be a universal set of gates, where $ G $ is a group of easily implementable gates (such as those that act transversally on the physical states), and $ \tau $ is a single gate outside of this group. Such collections of gates are known to support universal computations \cite{nielsen}. The search for codes that accept transversal action of a group of gates $G$ has been a frequent research topic in the literature, e.g., \cite{Moussa2016,BreuckmannBurton2022}. 
Sometimes such gate sets are called golden-gates, with a primary example of the form $ 2O+T $, where $ T $ is the square root of the phase gate and $ 2O $ is the binary octahedral group (a.k.a. the Clifford group). The authors of \cite{goldenGate} 
considered a universal set $ G+\tau $ that additionally minimizes the number of $\tau$ gates. They also defined another golden-gate set of the form $ 2I+\tau_{60},$
where $ 2I $ is the binary icosahedral group and $ \tau_{60} $ is another non-Clifford gate that they defined.}

\new{Permutation-invariant codes were linked to transversal gate sets in 
the recent paper \cite{exoticGates}, based on the results of \cite{gross}. Among other results, \cite{gross} constructed spin codes as representation of the group $2I$ that can be mapped onto permutation-invariant codes. For instance, \cite{gross} constructed a code spanned by the basis states
\begin{subequations}
  \begin{align}
    &\ket{c_0}=\sqrt{\frac{3}{10}}\;\Big|\frac{7}{2},\frac{7}{2}\Big\rangle+\sqrt{\frac{7}{10}}\;\Big|\frac72,-\frac{3}{2}\Big\rangle, \label{eq:2I codeC0}\\
    &\ket{c_1}=\sqrt{\frac{7}{10}}\;\Big|{\frac{7}{2},\frac{3}{2}}\Big\rangle-\sqrt{\frac{3}{10}}\;\Big|{\frac{7}{2},-\frac{7}{2}}\Big\rangle.\label{eq:2I codeC1}
  \end{align}
\end{subequations}
Following up on this work, the authors of \cite{exoticGates} defined a \textit{Dicke state mapping} $\mathscr{D}$ that converts a state of a spin-$j$ system into a permutation-invariant state on $n=2j$ qubits. It can be defined as follows:
\begin{align}  
    \mathscr{D} : \ket{j,m} \rightarrow \ket{D^{2j}_{j-m}}.
\end{align}
This mapping converts the logical gates of a spin code into the logical {\em transversal} gates of a permutation-invariant code. 
To link this line of work to our paper, observe that applying $\mathscr{D}$ to the spin code of \eqref{eq:2I codeC0}-\eqref{eq:2I codeC1}, we obtain exactly our
code $Q_{2,1,2,-}$ \eqref{eq:Q212}. Hence, this code admits the 2I group gates transversally \cite{exoticGates}.}

\new{Even more recently, paper \cite{kubischta2023notsosecret} introduced a family of 
permutation-invariant codes of distance $ 3 $ that admits transversal gates from $BD_{2b}$ (the binary dihedral group of degree $2b$). The group $BD_{2b}$ is a 
non-abelian subgroup of $ SU(2) $ of order $ 8b $ with generators
\begin{align*}
    X,Z,\begin{pmatrix}e^{-i\pi/2b}&0\\0&e^{i\pi/2b}\end{pmatrix}.
\end{align*}
For instance, $BD_2=\langle X,Z \rangle $, $BD_4=\langle X,Z,S \rangle $, and $BD_8=\langle X,Z,S,T \rangle $.
It is well known that $ [[2^{r+1}-1,1,3]] $ Reed-Muller codes implement the $BD_{2^r}$ group gates transversally. }

\new{\begin{proposition} Let  $b>0$ be an integer that is not of the form $ 2^r $ or $ 3(2^r) $.
    The codes in the family $Q_{3,1,2b-4,+}$ implement the group $ BD_{2b} $ transversally when $3\!\!\notdivides\!\! b$ and implement the group $ BD_{2b/3} $ transversally when $ 3|b $. The codes $ \cQ_{3,1,2^r-4,+} $  implement 
    the group $ BD_{2^r}  $ transversally for all integers $ r\geq 3 $.
\end{proposition}
This follows because the first code family in the proposition offers an alternative construction of the codes in Family 1 in \cite{kubischta2023notsosecret}, where the transversality properties are proved. The second code family in the proposition is the same as Family 2 in \cite{kubischta2023notsosecret}. }
  
\new{For example, the code $ \cQ_{3,1,4,+} $ of length $n=11$ with its basis codewords
\begin{align*}
    &\ket{\bfc_0} = \frac{\sqrt{5}}{4}\ket{D^{11}_0} + \frac{\sqrt{11}}{4}\ket{D^{11}_8},\\
    &\ket{\bfc_1} = \frac{\sqrt{11}}{4}\ket{D^{11}_3} + \frac{\sqrt{5}}{4}\ket{D^{11}_{11}}
\end{align*}
can correct one error and it implements the $T$ gate transversally. 
For comparison, the $[[15,1,3]]$ Reed-Muller code, which also has this property, is longer
than our construction.
Furthermore, the code $ \cQ_{3,1,12,+} $ with its codewords
\begin{align*}
    &\ket{\bfc_0} = \sqrt{\frac{13}{32}}\ket{D^{19}_0} + \sqrt{\frac{19}{32}}\ket{D^{19}_{16}},\\
    &\ket{\bfc_1} =  \sqrt{\frac{19}{32}}\ket{D^{19}_3} + \sqrt{\frac{13}{32}}\ket{D^{19}_{19}}.
\end{align*}
can correct one error, implements the $ \sqrt{T} $ gate transversally, and has better code parameters than the $[[31,1,3]]$ RM code that implements a tranversal $ \sqrt{T} $. These observations prompts us to inquire whether other codes in our family admit transversal logical gates.}


\section{Spontaneous Decay Errors}\label{sec:spontaneous}
In this section, we show that the codes constructed in Sec.~\ref{sec:new family}
correct errors of a different kind, arising from spontaneous photon emission.
\subsection{Basics of the amplitude damping channel}
This channel model arises from an approximation of noisy evolution in many physical systems. One of them is
the process in which an excited electron decays to its ground state, resulting in the emission of a photon. 
Say that the ground state is $\ket0$ and the excited state is $\ket 1$, and let the probability of decay be $p$, which is assumed to be small. Then, the behavior of this noise process on a single qubit system can be defined by the quantum channel
\begin{align}\label{eq:AD-1}
    {\cE}_p(\rho) = \bfA_0\rho\bfA_0^\dagger + \bfA_1\rho\bfA_1^\dagger,
\end{align}
where
\begin{align*}
    \bfA_0 = \begin{bmatrix}
        1&0\\
        0&\sqrt{1-p}
    \end{bmatrix},\quad \bfA_1 = \begin{bmatrix}
        0&\sqrt{p}\\
        0&0
    \end{bmatrix}.
\end{align*}
We clearly have $\bfA_0\ket0=\ket0,\bfA_0\ket 1=\sqrt{1-p}\ket1$ and $\bfA_1\ket0=0,\bfA_1\ket1=\sqrt p\ket0.$ Because of this, this channel model is called the \textit{amplitude damping channel} \cite{ADChuang}, \cite[Sec.4.4]{wilde},
and it forms a quantum analog of the classical $Z$-channel. The action of $\cE_p$ can be extended naturally to $n$-qubit systems by assuming that spontaneous decay affects independently each of the qubits in the superposition. We denote the $n$-qubit amplitude damping channel by ${\cE}_p^{\otimes n}.$ The Kraus set of this channel
has the form ${\cK}_{{\cE}_p^{\otimes n}} = \{ \otimes_{i=1}^n\bfK_i : \bfK_i \in \{ \bfA_0,\bfA_1 \}\}$. Let us further introduce the set 
of amplitude damping errors of multiplicity $t$, 
    \begin{align}\label{eq:truncatedKraus}
    \epsilon_{p,t} := \{ \bfK\in{\cK}_{{\cE}_p^{\otimes n}} : |\supp(\bfK)|\leq t \},
    \end{align}
where 
$ \bfK:= \otimes_{i=1}^n\bfK_i $ and $ \supp(\bfK) := \{ i \in \{1,2,\ldots, n\} :  \bfK_i = \bfA_1  \},$
calling it a \textit{truncated Kraus set} of $ {\cE}_p^{\otimes n}$ \cite{Ouyang_2013}.

In quantum coding theory, the problem of error correction is equivalent to minimizing the worst-case error of a code after the
recovery process. In other words, let $ {\cE} $ be a quantum channel, let $ \cC $ be a quantum code, and let ${\cR}$ be the recovery
operator that corrects errors introduced by the channel. Then, the worst-case error is 
\begin{align*}
    E_{{\cE},\cC}({\cR}) := \max_{\rho\in S(\cC)}\left[ 1 - F(\rho, {\cR}\circ {\cE}  ) \right],
\end{align*}
where $ S(\cC)=\{ \rho\in S(\complex^{q}) : \sum_{ \ket{c}\in{\cB}} \bra{c}\rho\ket{c} = 1  \} $ (here $ {\cB} $ is an orthonormal basis of $ \cC $ and $q\geq 2$ is an integer), and $  F(\rho, {\cR}\circ{\cE}) $ is the \textit{entanglement fidelity}, defined as 
\begin{align*}
     F(\rho, {\cR}\circ {\cE}  ) := \sum_{\bfA \in {\cK}_{{\cR}\circ{\cE} }}|\Tr(\bfA\rho)|^2, 
\end{align*}
where ${\cK}_{{\cR}\circ{\cE} }$ is the Kraus set of the channel $ {\cR}\circ{\cE} $ \cite{schumacher}, \cite[p.228]{wilde}.
The fidelity is a way to measure how close the recovered density matrix ${\cR}\circ{\cE}(\rho)$ is to the original matrix $\rho$. 

With this, the error correction problem can be stated as the following min-max problem: 
\begin{align*}
    \inf_{{\cR}}E_{{\cE},\cC}({\cR}) =  \inf_{{\cR}}\max_{\rho\in S(\cC)}\left[ 1 - F(\rho, {\cR}\circ {\cE}  ) \right].
\end{align*}
Following \cite{ouyangPI}, we say that the code corrects $ t $ amplitude damping errors if there exist some positive constants $ A $ and $ p_0 $ such that
\begin{align}\label{eq:tADBound}
     \inf_{{\cR}}E_{{\cE}_p^{\otimes n},\cC}({\cR}) \leq Ap^{t+1}
\end{align}
holds for all $ p\in[0,p_0] $. 

In this section, we quantify the error correction properties of the codes $\mathcal{Q}_{g,m,\delta,\epsilon}$. Our main result here is given in the following theorem.
\begin{theorem}\label{theorem:ADCorrection}
Let $ t $ be a nonnegative integer. Let $ g\geq t+1 $, $ m\geq \ceil{\frac{3t}{2}} $, $\delta\geq t$, and \new{$\epsilon =\pm 1$}. Then the code $\cQ_{m,l,\delta,\epsilon}$ corrects $t$ amplitude-damping errors.
\end{theorem}

The general tool for proving error correction in the sense of \eqref{eq:tADBound} is given in the following theorem.
\begin{theorem}\label{theorem:OuyangTheorem10}
    (\cite{ouyangPI}, Theorem 10) Let $ \epsilon $ be a truncated Kraus set, and $ \bfA,\bfB \in \epsilon $. Let $ \cC $ be a code with an orthonormal basis $ {\cB} $. If $ \eta = \frac{(|\epsilon|-1)|\epsilon|^2\Delta}{\lambda_{\min}(\bfM)}$, then
    \begin{align}\label{eq:theorem10}
        \inf_{{\cR}}E_{{\cE},\cC}({\cR}) \leq 1 - \frac{\Tr\bfM - |\epsilon|^2\Delta}{1+\eta},
    \end{align}
    where 
    \begin{align}\label{eq:matrixM}
        \bfM := \sum_{\bfA,\bfB\in\epsilon}m_{\bfA,\bfB}\ket{\bfA}\bra{\bfB}, \quad m_{\bfA,\bfB}:=\frac{1}{|{\cB}|}\sum_{\ket{c_i}\in{\cB}}\bra{c_i}\bfA^\dagger\bfB\ket{c_i},
    \end{align}
    \begin{align}\label{eq:DefDelta}
        \Delta := \max_{\bfA,\bfB}\max_{i}\bfM_{\bfA,\bfB}^{ii} + \left( |{\cB}|-1 \right)\max_{\bfA,\bfB}\max_{\substack{i,j \\ i\neq j}}\bfM_{\bfA,\bfB}^{ij},
    \end{align}
    where   \begin{align}\label{eq:DefM_AB}
        \bfM_{\bfA,\bfB} := \sum_{\ket{c_i},\ket{c_j}\in{\cB}}\left( \bra{c_i}\bfA^\dagger\bfB\ket{c_j}-m_{\bfA,\bfB}\delta_{\ket{c_i},\ket{c_j}} \right)\ket{c_i}\bra{c_j},
    \end{align}
and $\bfM^{ij}$ denotes the matrix element indexed by $i,j.$
\end{theorem}

Let $ {\cE}_p^{\otimes n} $ be the amplitude damping channel with decay probability $p$ and let $ \epsilon_{p,t} \subset {\cK}_{ {\cE}_p^{\otimes n} } $ be the truncated Kraus set as defined in \eqref{eq:truncatedKraus}. The following lemma provides a lower bound for the trace of matrix $ \bfM $ in \eqref{eq:matrixM}:
\begin{lemma}\label{Lemma:LowerBoundTrace}
    (\cite{ouyangPI}, Lemma 11) Let $ p>0 $ be a real number. Then 
    \begin{align*}
        \Tr\bfM \geq \lambda_{\min}\Big( \sum_{\bfA\in\epsilon_{p,t} }\bfA^\dagger\bfA \Big)\geq 1-\binom{n}{t+1}p^{t+1},
    \end{align*}
where $\lambda_{\min}$ is the smallest eigenvalue of $\bfM.$
\end{lemma}

Define $ a:=|\supp(\bfA)|$, $ b:=|\supp(\bfB)| $ and $ c:=|\supp(\bfA)\cup\supp(\bfB)| - a $. The following lemma describes the action of amplitude damping errors on Dicke states.
\begin{lemma}
    (\cite{ouyangPI}, Lemma 13) Let $ \bfA,\bfB\in\epsilon_{p,t} $. Then
    \begin{align*}
        \bra{D^n_w}\bfA^\dagger\bfB\ket{D^n_w} = p^a(1-p)^{w-a}\frac{\binom{n-c-a}{w-a}}{\binom{n}{w}}\delta_{a,b}.
    \end{align*}
\end{lemma}
We will use this equality in the form
\begin{align}\label{eq:DickeStatesADErrors}
    \bra{D^n_w}\bfA^\dagger\bfB\ket{D^n_w} = \sum_{k=0}^{w}(-1)^{k-a}\frac{\binom{w-a}{k-a}\binom{n-c-a}{w-a}}{\binom{n}{w}}p^k\delta_{a,b}.
\end{align}
Let $ \{\ket{c_+}, \ket{c_-}\}$ be the Hadamard basis of the code $ \cQ_{g,m,\delta,\epsilon} $, namely $ \ket{c_+}=\frac{\ket{c_0}+\ket{c_1}}{\sqrt{2}} $ and $ \ket{c_-}=\frac{\ket{c_0}-\ket{c_1}}{\sqrt{2}} $. The following lemma provides an upper bound for $ \Delta $ in \eqref{eq:DefDelta}.
\begin{lemma}\label{Lemma:UpperBoundDelta}
    Let $ \bfA,\bfB \in \epsilon_{p,t} $. \remove{Let $ p_1\in(0,1) $ be a real number. Then, for all $ p<p_1 $.} Then, for all $ p\in(0,1),$
    \begin{align*}
        \Delta\leq Cp^{2m-t+1},
    \end{align*}
    where\footnote{$[p^k](\cdot)$ refers to the coefficient of $p^k$ in the expansion of the expression
    in the parentheses in powers of $p$.}
    \begin{align}\label{eq:DefC}
        C:=\max_{\bfA,\bfB}\sum_{k\geq 2m-t+1}|[p^k]\bra{c_+}\bfA^\dagger\bfB\ket{c_-}|. 
    \end{align}
 \end{lemma}
\begin{proof}
We start with writing $ m_{\bfA,\bfB} $ in \eqref{eq:matrixM} explicitly:
\begin{align*}
    m_{\bfA,\bfB} = \frac{1}{2}\bra{c_0}\bfA^\dagger\bfB\ket{c_0}+\frac{1}{2}\bra{c_1}\bfA^\dagger\bfB\ket{c_1}.
\end{align*}
Observe that $ \bra{c_0}\bfA^\dagger\bfB\ket{c_1} = \bra{c_1}\bfA^\dagger\bfB\ket{c_0} = 0 $ since $ g\geq t+1 $. Hence, the matrix $ \bfM_{\bfA,\bfB} $ in \eqref{eq:DefM_AB} can be written as
\begin{align*}
    \bfM_{\bfA,\bfB} &= \frac{\bra{c_0}\bfA^\dagger\bfB\ket{c_0} - \bra{c_1}\bfA^\dagger\bfB\ket{c_1} }{2}\left( \ket{c_0}\bra{c_0}-\ket{c_1}\bra{c_1} \right)\\
    &=\bra{c_+}\bfA^\dagger\bfB\ket{c_-}\left( \ket{c_0}\bra{c_0}-\ket{c_1}\bra{c_1} \right).
\end{align*}
We obtain
\begin{align}\label{eq:DeltaExplicit}
    \Delta = \max_{\bfA,\bfB}| \bra{c_+}\bfA^\dagger\bfB\ket{c_-} |.
\end{align}
The $ \cQ_{g,m,\delta,\epsilon} $ code in Construction \ref{constructionGmdelta} can be written in the Hadamard basis as follows:
\begin{align*}
    &\ket{c_+} = \sum_{l=0}^m\frac{\gamma b_l}{\sqrt{2}}\ket{D^n_{gl}} +\epsilon\sum_{l=0}^m\frac{\gamma b_l}{\sqrt{2}}(\epsilon)^l\ket{D_{n-gl}^n}\\
    &\ket{c_-} = \sum_{l=0}^m\frac{\gamma b_l}{\sqrt{2}}(-1)^l\ket{D^n_{gl}} -\epsilon \sum_{l=0}^m\frac{\gamma b_l}{\sqrt{2}}(-\epsilon)^l\ket{D_{n-gl}^n}.
\end{align*}
Let $ f(l)=(\gamma b_l)/\sqrt{2} $. Then, the inner product
\begin{multline*}
   \bra{c_+}\bfA^\dagger\bfB\ket{c_-} = \sum_{l=0}^m\sum_{l^\prime=0}^mf(l)f(l^\prime)(-1)^{l^\prime}\bra{D^n_{gl}}\bfA^\dagger\bfB\ket{D^n_{gl^\prime}}\\
   -\epsilon \sum_{l=0}^m\sum_{l^\prime=0}^mf(l)f(l^\prime)(-\epsilon)^{l^\prime}\bra{D^n_{gl}}\bfA^\dagger\bfB\ket{D^n_{n-gl^\prime}}\\  +\epsilon \sum_{l=0}^m\sum_{l^\prime=0}^mf(l)f(l^\prime)(\epsilon)^l(-1)^{l^\prime}\bra{D^n_{n-gl}}\bfA^\dagger\bfB\ket{D^n_{gl^\prime}} \\-\epsilon^2\sum_{l=0}^m\sum_{l^\prime=0}^mf(l)f(l^\prime)(\epsilon)^{l}(-\epsilon)^{l^\prime}\bra{D^n_{n-gl}}\bfA^\dagger\bfB\ket{D^n_{n-gl^\prime}}.
\end{multline*}   
Observe that $ \bra{D^n_{gl}}\bfA^\dagger\bfB\ket{D^n_{n-gl^\prime}} = \bra{D^n_{n-gl}}\bfA^\dagger\bfB\ket{D^n_{gl^\prime}} =0 $ since the weight of any state can decrease by at most 
$ t $ upon applying $ \bfA,\bfB,$ and $ n-gl-t \neq gl^\prime $ for any $ l,l^\prime$. To see this, recall 
that $ l+l^\prime\leq 2m $ and $ \delta\geq t $, which yields $ g(l+l^\prime)+t \leq 2gm + t < 
2gm+\delta+1 =n $. Furthermore, the inner products $ \bra{D^n_{gl}}\bfA^\dagger\bfB\ket{D^n_{gl^\prime}} $ and 
$ \bra{D^n_{n-gl}}\bfA^\dagger\bfB\ket{D^n_{n-gl^\prime}} $ are zero unless $ l=l^\prime,$ 
since $g\geq t+1.$ Therefore, recalling $ \epsilon^2=1 $, we obtain 
\begin{align}\label{eq:innerProduct}
    \bra{c_+}\bfA^\dagger\bfB\ket{c_-} = \sum_{l=0}^mf(l)^2(-1)^l\left( \bra{D^n_{gl}}\bfA^\dagger\bfB\ket{D^n_{gl}} - \bra{D^n_{n-gl}}\bfA^\dagger\bfB\ket{D^n_{n-gl}} \right).
\end{align}
By combining \eqref{eq:DickeStatesADErrors} with \eqref{eq:innerProduct} and interchanging the order of summation, we obtain
\begin{align}
    \bra{c_+}\bfA^\dagger\bfB\ket{c_-} &= \sum_{k \geq 0}(-1)^{k-a}p^k  \sum_{l=0}^m(-1)^l\frac{f(l)^2}{\binom{n}{gl}}\Big\{{\binom{gl-a}{k-a}\binom{n-c-a}{gl-a}}\notag\\ &
    \hspace*{.3in}-{\binom{n-gl-a}{k-a}\binom{n-c-a}{n-gl-a}}\Big\}  + \mathcal{O}(p^{gm+1}). \label{eq:innerProduct2}
\end{align}
Here we considered Kraus operators $ \bfA $ and $ \bfB $ such that $ |\supp(\bfA)| = |\supp(\bfB)| $. Recall that $ c,a\leq t,$ which together with Lemma \ref{lemma:C3} below implies that the inner sum in \eqref{eq:innerProduct2} is zero for all $ k\leq 2m-t.$ Now using \eqref{eq:DeltaExplicit} and \eqref{eq:innerProduct2} completes the proof.
\end{proof}

We will borrow the following lemma from \cite{ouyangPI} with a small change.
\begin{lemma}[\rm \cite{ouyangPI}, Lemma 15]\label{Lemma:LowerBoundLambdaMin}
     Let $ \bfA\in\epsilon_{p_1,t} $ and $ p_1<1/2 $ be a real number. Let
    \begin{align*}
        p_0 = n^{-t/(2m-2t+1)}\left( \frac{D}{2C} \right)^{1/(2m-2t+1)},
    \end{align*}
    where $ C $ is given by \eqref{eq:DefC} and
    \begin{align*}
        D:=\min_{\bfA}\min\{\bra{c_0}\bfA^\dagger\bfA\ket{c_0}, \bra{c_1}\bfA^\dagger\bfA\ket{c_1}   \}. 
    \end{align*}
    Suppose that $ p_0\leq p_1 $. Then, for all $ p < p_0 $,
    \begin{align*}
        \lambda_{\min}(\bfM)\geq \frac{D}{2}p^t.
    \end{align*}
\end{lemma}

{\em Proof of Theorem~\ref{theorem:ADCorrection}:} 
By using Theorem \ref{theorem:OuyangTheorem10} and Lemmas \ref{Lemma:LowerBoundTrace}, \ref{Lemma:UpperBoundDelta}, and \ref{Lemma:LowerBoundLambdaMin}, we can bound the worst case error for the amplitude damping channel on $ n $ qubit as follows:
\begin{align}\label{eq:FinalUpperBound}
    \inf_{{\cR}}E_{{\cE}_p^{\otimes n},\cC}({\cR}) \leq 1 - \frac{1-\binom{n}{t+1}p^{t+1}-|\epsilon_{p,t}|^2Cp^{2m-t+1}}{1+\frac{2C|\epsilon_{p,t}|^2(|\epsilon_{p,t}|-1)}{D}p^{2m-2t+1}}.
\end{align}
Note that if $ m \geq \ceil{\frac{3t}{2}} $ holds, then the upper bound in \eqref{eq:FinalUpperBound} converges to zero in the rate of $ t+1 $ as $ p $ goes to zero. This proves the theorem. \hfil\qed

\vspace*{.1in} For example, consider the code $\cQ_{3,3,2,-}$ with basis states
\begin{align*}
    &\ket{c_0} = \frac{1}{8}\left( \ket{D_0^{21}} + \sqrt{21}\ket{D_6^{21}} + \sqrt{35}\ket{D_{12}^{21}} +\sqrt{7}\ket{D_{18}^{21}}\right),\\
    &\ket{c_1} = \frac{1}{8}\left( \sqrt{7}\ket{D_3^{21}} + \sqrt{35}\ket{D_9^{21}} - \sqrt{21}\ket{D_{15}^{21}} -\ket{D_{21}^{21}}\right),
\end{align*}
As shown above, this code is of length {21} and it corrects 2 amplitude damping errors. 

To compare the codes $\cQ_{g,m,\delta,\epsilon}$ with the existing constructions of {permutation-invariant} codes that correct amplitude damping errors, we note that
the shortest code in Ouyang's {\em gnu} family that corrects $t$ amplitude damping errors has length 
$(t+1)(3t+1)+t.$ At the same time, taking $(g,m,\delta,\epsilon) = (t+1,\ceil{\frac{3t}{2}}, t,\mp 1)$, we obtain a code $\cQ_{g,m,\delta,\epsilon}$ 
of length $(t+1)(1+2\ceil{\frac{3t}{2}})$. Thus, for an even $t$ our construction requires $t$ fewer
physical qubits than the best {permutation-invariant} codes known previously. 
At the same time, for odd $t$ it needs 1 more physical qubit compared to the {\em gnu} codes.

The next lemma was referenced toward the end of the proof of Lemma~\ref{Lemma:UpperBoundDelta}.
\begin{lemma}\label{lemma:C3}
    Let $ a,c,k,g,n,m,t $ be nonnegative integers. For all $ n>2gm $, $ k\leq 2m -t $, $ c\leq a\leq t\leq m $, 
    \begin{align}\label{eq:E2}
        \sum_{l=0}^{m}(-1)^l\frac{\binom{m}{l}}{\binom{n/g-l}{m+1}}\Big[\frac{\binom{gl-a}{k-a}\binom{n-(c+a)}{gl-a}}{\binom{n}{gl}}-\frac{\binom{n-gl-a}{k-a}\binom{n-(c+a)}{n-gl-a}}{\binom{n}{gl}}\Big]=0.
        \end{align}
\end{lemma}
\begin{proof}
    Since $\binom{n-c-a}{gl-a}\binom{gl-a}{k-a}=\binom{n-c-a}{k-a}\binom{n-c-k}{gl-k}$ and $\binom{n-c-a}{n-gl-a}\binom{n-gl-a}{k-a}=
    \binom{n-c-a}{k-a}\binom{n-c-k}{n-gl-k}$, the left-hand side of \eqref{eq:E2} can be rewritten as
    $$
    \binom{n-c-a}{k-a}\sum_{l=0}^m (-1)^l\frac{\binom ml}{\binom{n/g-l}{m+1}\binom{n}{ gl}}\Big[
    \binom{n-c-k}{gl-k}-\binom{n-c-k}{gl-c}\Big],
    $$
which is zero by Lemma~\ref{lemma:C1}.
\end{proof}

\section{Generalization of the Pollatsek--Ruskai Conditions}\label{sec:PR conditions}
In \cite[Thm.~1]{ruskai-polatsek} the authors formulated necessary and sufficient conditions for permutation-invariant codes of a specific form \eqref{eq:codewords} to correct a single error (and some double errors). In this section we will generalize
their conditions to extend to arbitrary patterns of $t$ errors for all $t\ge 1.$
The permutation-invariant code $\cC_n$ of \cite{ruskai-polatsek} has logical codewords
    \begin{equation}
        \ket{\bfc_0} = \sum_{l=0}^{(n-1)/2}q_{2l}\sqrt{\binom{n}{2l}}\ket{D^n_{2l}} \quad \text{and} \quad \ket{\bfc_1} = \sum_{l=0}^{(n-1)/2}q_{n-2l-1}\sqrt{\binom{n}{2l+1}}\ket{D^n_{2l+1}}, 
        \label{eq:PRcode}
    \end{equation}
where the states are not normalized, and where $n$ is assumed to be an odd integer.
The conditions for the code $\cC_n$ to correct $t$ qubit errors have the following form.
\begin{proposition}\label{corollaryRK}
Let $a,b$ be nonnegative integers. For real coefficients $q_0,q_2,\ldots,q_{n-1},$ not all zero, conditions {\rm(C3)} and 
{\rm(C4)}
for the code $\cC_n$ can be equivalently stated as 
    \begin{alignat*}{2} 
    &\text{{\rm (D1)}} \text{ For all even $a$ and odd $b$,} 
                                  &&\hspace*{-.15in}{a,b\leq 2t}, \\       
        && &\hspace*{-1in}\sum_{k=0}^{\frac{n-1}{2}-t}\binom{n-2t}{2k}q_{2k+a}q_{n-2k-b}=0; \quad\\
      & \text{{\rm (D2)}} \text{ For all even $a$ and $b$, $ a\leq b $, } &&a+b < 2t,\\
          & & &\hspace*{-1in}\sum_{k=0}^{\frac{n-1}{2}-t}\binom{n-2t}{2k}\left(q_{2k+a}q_{2k+b} - q_{2k+2t-a}q_{2k+2t-b}\right) = 0; \\
       & \text{{\rm (D3)}} \text{ For all odd $ a $ and $ b $, $ a\leq b $, } && a+b < 2t, \\ 
       & & &\hspace*{-1in}\sum_{k=0}^{\frac{n-1}{2}-t}\binom{n-2t}{2k}\left(q_{n-2k-a}q_{n-2k-b} - q_{n-2k-2t+a}q_{n-2k-2t+b}\right) = 0.
    \end{alignat*}
\end{proposition}
\begin{proof} Observe that the coefficient vectors $ \underline{\alpha} $ and $ \underline{\beta} $ in \eqref{eq:PRcode} are
\begin{align*}
    &\alpha_j = \sum_{l=0}^{(n-1)/2}q_{j}\sqrt{\binom{n}{j}}\delta_{j,2l},\\
    &\beta_j = \sum_{l=0}^{(n-1)/2}q_{n-j}\sqrt{\binom{n}{j}}\delta_{j,2l+1}.
\end{align*}
We begin with writing the term
\begin{align*}
    \alpha_{j+a}\beta_{j+b} &= \sum_{l=0}^{(n-1)/2}\sum_{l^\prime=0}^{(n-1)/2}q_{j+a}q_{n-j-b}\sqrt{\binom{n}{j+a}\binom{n}{j+b}}\delta_{2l-a,2l^\prime+1-b}\delta_{j,2l-a}.
\end{align*}
We see that the product $ \alpha_{j+a}\beta_{j+b} \ne 0$ only if $ 2l^\prime = 2l + b - a -1 $. For this to hold, since both $ l $ and $ l^\prime $ are integers, the numbers $a$ and $b$ must be of different parity. 
By symmetry, the cases (odd,even) and (even,odd) lead to the same expression.
Thus, it suffices to consider only one of them, say that of even $a$ and odd $b$ where $ 0\leq a,b \leq 2t $.
Condition $(C3)$ yields
\begin{align*}
    \sum_{j=0}^n\frac{\binom{n-2t}{j}}{\sqrt{\binom{n}{j+a} \binom{n}{j+b}}} \alpha_{j+a}\beta_{j+b} &= \sum_{l=0}^{(n-1)/2}\sum_{j=0}^n\binom{n-2t}{j}q_{j+a}q_{n-j-b}\delta_{j,2l-a}\\
    &=\sum_{l=0}^{(n-1)/2}\binom{n-2t}{2l-a}q_{2l}q_{n-2l+a-b}\\
    &=\sum_{k=0}^{\frac{n-1}{2}-t}\binom{n-2t}{2k}q_{2k+a}q_{n-2k-b}.
\end{align*}
Hence, for the code $\cC_n$, conditions (C3) and (D1) are equivalent. Now, let us write the terms
\begin{align*}
    \alpha_{j+a}\alpha_{j+b} &= \sum_{l=0}^{(n-1)/2}\sum_{l^\prime=0}^{(n-1)/2}q_{j+a}q_{j+b}\sqrt{\binom{n}{j+a}\binom{n}{j+b}}\delta_{2l-a,2l^\prime-b}\delta_{j,2l-a} 
\end{align*}
and
\begin{align*}
    \beta_{j+a}\beta_{j+b} &= \sum_{l=0}^{(n-1)/2}\sum_{l^\prime=0}^{(n-1)/2}q_{n-j-a}q_{n-j-b}\sqrt{\binom{n}{j+a}\binom{n}{j+b}}\delta_{2l+1-a,2l^\prime+1-b}\delta_{j,2l+1-a}. 
\end{align*}
The products $ \alpha_{j+a}\alpha_{j+b}  $ and $ \beta_{j+a}\beta_{j+b} $ are nonzero only when $2l = 2l^\prime + a- b$, implying that $a$ and $b$ have the same parity. 
We start with both being even, obtaining
\begin{align*}
        \sum_{j=0}^n\frac{\binom{n-2t}{j}}{\sqrt{\binom{n}{j+a} \binom{n}{j+b}}}\alpha_{j+a}\alpha_{j+b} &= \sum_{l=0}^{(n-1)/2} \sum_{j=0}^n \binom{n-2t}{j}q_{j+a}q_{j+b}\delta_{j,2l-a}\\
        &=\sum_{l=0}^{(n-1)/2}\binom{n-2t}{2l-a}q_{2l}q_{2l-a+b}\\
        &=\sum_{k=0}^{\frac{n-1}{2}-t}\binom{n-2t}{2k}q_{2k+a}q_{2k+b},
\end{align*}
and
\begin{align*}
        \sum_{j=0}^n\frac{\binom{n-2t}{j}}{\sqrt{\binom{n}{j+a} \binom{n}{j+b}}}\beta_{j+a}\beta_{j+b} &= \sum_{l=0}^{(n-1)/2} \sum_{j=0}^n \binom{n-2t}{j}q_{n-j-a}q_{n-j-b}\delta_{j,2l+1-a} \\
        &=\sum_{l=0}^{(n-1)/2}\binom{n-2t}{2l+1-a}q_{n-2l-1}q_{n-2l-1+a-b}\\
        &=\sum_{k=0}^{\frac{n-1}{2}-t}\binom{n-2t}{2k}q_{2k+2t-a}q_{2k+2t-b}.
\end{align*}
In the end we obtain
\begin{align}\label{exp1}
     \sum_{j=0}^n\frac{\binom{n-2t}{j}}{\sqrt{\binom{n}{j+a} \binom{n}{j+b}}}\left(\alpha_{j+a}\alpha_{j+b}-\beta_{j+a}\beta_{j+b}\right) = \sum_{j=0}^{\frac{n-1}{2}-t}\binom{n-2t}{2k}\left(q_{2k+a}q_{2k+b} - q_{2k+2t-a}q_{2k+2t-b}\right).
\end{align}
Likewise, for $a$ and $b$ odd, we compute
\begin{multline}\label{exp2}
    \sum_{j=0}^n\frac{\binom{n-2t}{j}}{\sqrt{\binom{n}{j+a} \binom{n}{j+b}}}\left(\alpha_{j+a}\alpha_{j+b}-\beta_{j+a}\beta_{j+b}\right) \\=\sum_{k=0}^{\frac{n-1}{2}-t}\binom{n-2t}{2k}\left(q_{n-2k-a}q_{n-2k-b} - q_{n-2k-2t+a}q_{n-2k-2t+b}\right).
\end{multline}
Observe that \eqref{exp1} and \eqref{exp2} are invariant under the exchange $a\leftrightarrow b$, and they are trivially zero when $ a+b=2t $. Furthermore, up to the sign, they have the same values in the regions
$ a+b>2t $ and $ a+b<2t $. Since conditions (D2) and (D3) require \eqref{exp1} and \eqref{exp2} to be zero,
they can result in different values only if $ a\leq b $ and $ a+b<2t $. 
This shows that conditions (D2) and (D3) together are equivalent to (C4). 
The proof is now complete.
\end{proof}

For $ t=1 $, conditions (D1), (D2), and (D3) are equivalent to the conditions 
for error correction in \cite[Thm.~1]{ruskai-polatsek}. Here we have extended them for any $t$, and thus
the code $\cC_n$ corrects all $t$ qubit errors if and only if it has real coefficients that satisfy conditions
(D1)-(D3).

We end this section with a remark concerning code construction. Observe that condition (D1) produces
$t(t+1) $ quadratic equations for the coefficients, and condition (D2) and (D3) together yield $t(t+1)/2$ more.
Therefore, to construct a $t$-error-correcting code that satisfies Proposition~\ref{corollaryRK}, 
we need to solve a system of $ 3t(t+1)/2 $ quadratic equations with respect to $(n+1)/2 $ real unknowns, where $ n $ is the code length. Generally, such a system is more likely than not to be incompatible, except for trivial solutions
if the length $ n<3t^2+3t-1 $ since it would be over-determined.
For codes of length $ n\geq3t^2+3t-1 $, there is also no guarantee of a non-trivial solution, and even
attempting to solve the system by computer is a non-trivial task. Along this path, for $ t=1 $, Pollatsek and
Ruskai \cite{ruskai-polatsek} showed that there is no code of length $5$ that satisfies these conditions.
The shortest permutation-invariant code for the single errors they obtained has length 7. 
For $ t=2 $, it can be shown that there is no code of length shorter than $ 19 $. 
For $ n\geq 19 $, one can try to solve a set of $ 9 $ quadratic equations to find an explicit code for double errors.

As an example, let us examine the case where $ t = 1 $, and $ n= 7 $. By (D1) the options for $(a,b)$ are $ (0,1) $ 
or $ (2,1) $, which yields the following equations:
\begin{align*}
    0&=\binom{5}{0}q_0q_6 + \left(\binom{5}{2}+\binom{5}{4}\right)q_2q_4,  \\
    0&=\left(\binom{5}{0}+\binom{5}{4} \right)q_2q_6 + \binom{5}{2}q_4^2.
\end{align*}
From (D2) for $ (a,b)=(0,0) $ we have
\begin{align*}
    0 = \binom{5}{0}\left(q_0^2-q_2^2\right)+\binom{5}{2}\left(q_2^2-q_4^2\right) + \binom{5}{4}\left( q_4^2-q_6^2 \right).
\end{align*}
As far as (D3) is concerned, there are no pairs $(a,b)$, both odd, such that $a+b<2t$, so this condition is vacuous and can be ignored. Performing simplifications, we obtain the following set of equations
\begin{align*}
    3q_2q_6 + 5q_4^2 = 0&\\
    q_0q_6 + 15q_2q_4 = 0\\
    q_0^2 + 9q_2^2 - 5q_4^2 - 5q_6^2 = 0.    
\end{align*}
Solving this system, we obtain the ((7,2,3)) permutation-invariant code of \cite{ruskai-polatsek}. This is to be expected since when $t=1,$ conditions (D1)-(D3) are equivalent to the conditions in \cite[Thm.~1]{ruskai-polatsek}. 
At the same time, conditions (D1)-(D3) are sufficient for the existence of codes that correct any number $t$ of errors, although their use becomes more difficult as $t$ increases.
For instance, for
$t=2$ and $n=19$, conditions (D1)-(D3) give rise to 9 quadratic equations. Solving them by computer, we determine that there exists a $((19,2,5))$ permutation-invariant code, and one choice of the coefficients $q_{2i},i=0,\dots,9$ has the form
{\small
\begin{gather*}
q_0=1, q_{2}=0.0477572, q_{4}=-0.0267249, q_{6}=-0.00506367, 
  q_{8}=0.00332914, q_{10}=0.00527235,\\ q_{12}=-0.000947223, 
  q_{14}=0.0152707, q_{16}=0.00888631, 
  q_{18}=0.32678.
\end{gather*}  
}  
These numbers are approximations of the solution, produced by {\tt msolve}, a C library for solving systems of polynomial equations  \cite{msolve}. Its output is an interval for each of the variables, where the solution is actually contained. These numbers
were also verified by Wolfram Mathematica. 

  \remove{
  {q[2]=0.0895044, q[4]=0.0265365, 
  q[6]=-0.00749379, q[8]=0.00481188, q[10]=-0.00464018, 
  q[12]=-0.00848491, q[14]=-0.00575057, q[16]=0.0440988, 
  q[18]=-0.340289}, {q[2]=-0.0895044, q[4]=0.0265365, 
  q[6]=0.00749379, q[8]=0.00481188, q[10]=0.00464018, 
  q[12]=-0.00848491, q[14]=0.00575057, q[16]=0.0440988, 
  q[18]=0.340289}, {q[2]=0.0441768, q[4]=0.00569365, 
  q[6]=0.00748978, q[8]=-0.00540083, q[10]=-0.00166005, 
  q[12]=-0.00170267, q[14]=-0.00530244, q[16]=-0.0423825, 
  q[18]=0.223268}, {q[2]=0.0441768, q[4]=0.00569365, 
  q[6]=0.00748978, q[8]=-0.00540083, q[10]=-0.00166005, 
  q[12]=-0.00170267, q[14]=-0.00530244, q[16]=-0.0423825, 
  q[18]=0.223268}, {q[2]=-0.0441768, q[4]=0.00569365, 
  q[6]=-0.00748978, q[8]=-0.00540083, q[10]=0.00166005, 
  q[12]=-0.00170267, q[14]=0.00530244, q[16]=-0.0423825, 
  q[18]=-0.223268}, {q[2]=0.0460251, q[4]=0.0168968, 
  q[6]=-0.00418959, q[8]=0.00470797, q[10]=0.0016692, 
  q[12]=-0.000325222, q[14]=-0.00730742, q[16]=-0.0362696, 
  q[18]=0.264061}, {q[2]=-0.0460251, q[4]=0.0168968, 
  q[6]=0.00418959, q[8]=0.00470797, q[10]=-0.0016692, 
  q[12]=-0.000325222, q[14]=0.00730742, q[16]=-0.0362696, 
  q[18]=-0.264061}}

\section{Concluding remarks}
In this paper, we introduced the necessary and sufficient conditions for a permutation-invariant code to
correct arbitrary $ t $ errors. We also presented a family of permutation-invariant codes that
can be defined explicitly using parameters $ g,m,\delta,\epsilon $. By adjusting these parameters, one can show
that the proposed codes correct arbitrary $t$ Pauli errors, $t$ amplitude damping errors, or $ t $ deletion errors. 
The minimum length of our codes is smaller than in the
previous explicit permutation-invariant code constructions. Since any permutation-invariant state must
necessarily be a ground state of the ferromagnetic Heisenberg model in the absence of an external magnetic field, 
the proposed codes are also suitable for a range of applications discussed in \cite{ouyangStorage,ouyangSensors,ouyanggCommunication}.

\section{Acknowledgment} The research of A. Barg was partially supported by NSF grants CCF-2110113 (NSF-BSF), CCF-2104489, and CCF-2330909.

\vspace*{.3in}
\appendix
\section{Proofs of combinatorial lemmas from Sec.~\ref{sec:new family}}
\label{sec:Proofs}

\begin{proof}[Proof of Lemma~\ref{lemma:C1}] 
First, let us notice that
\begin{align*}
    \frac{\binom{m}{l}}{\binom{n/g-l}{m+1}} &= \binom{n/g}{l}\binom{n/g-l-m-1}{m-l}\frac{m!(m+1)!\Gamma\left( n/g-2m \right)}{\Gamma\left( n/g+1 \right)}
\end{align*}
\new{The following {\em negation relation} is obtained directly by definition:
$
      \binom{x}{r}=(-1)^r\binom{r-x-1}{r}.
      $
Negating the second binomial on the right in the above equality, we obtain}
\begin{align*}
    = \frac{\binom{m}{l}}{\binom{n/g-l}{m+1}} 
    &= (-1)^{m-l}\binom{n/g}{l}\binom{2m-n/g}{m-l}\frac{m!(m+1)!\Gamma\left( n/g-2m \right)}{\Gamma\left( n/g+1 \right)}.
\end{align*}
Similarly, we obtain
\begin{align*}
    \frac{\binom{n-r}{gl-a}}{\binom{n}{gl}} = &\binom{gl}{a}\binom{n-gl}{r-a}\frac{a!(r-a)!(n-r)!}{n!},\\
    \frac{\binom{n-r}{gl-r+a}}{\binom{n}{gl}} = &\binom{gl}{r-a}\binom{n-gl}{a}\frac{a!(r-a)!(n-r)!}{n!}.   
\end{align*}
Hence, identity \eqref{eq:E1} is equivalent to the following identity:
\begin{align}
    \sum_{l=0}^m\binom{n/g}{l}\binom{2m-n/g}{m-l}\binom{gl}{a}\binom{n-gl}{r-a} = \sum_{l=0}^m\binom{n/g}{l}\binom{2m-n/g}{m-l}\binom{gl}{r-a}\binom{n-gl}{a}.
\label{eq:two_sums}
\end{align}
To prove \eqref{eq:two_sums} at once for all $a,r$ satisfying $0\leq a\leq r\leq 2m$, we first convert it into a power series identity. To do so, we multiply it by $x^a$, sum over $a=0,1,\dots,r$, and note that
\begin{align*}
      \sum_{a=0}^{r}\binom{gl}{a}\binom{n-gl}{r-a}x^a = \left[y^{r}\right]\ \left( 1+xy \right)^{gl}\left( 1+y \right)^{n-gl}
\end{align*}
and
\begin{align*}
      \sum_{a=0}^{r}\binom{gl}{r-a}\binom{gl}{a}x^a = \left[y^{r}\right]\ \left( 1+y \right)^{gl}\left( 1+xy \right)^{n-gl},
\end{align*}
where $\left[y^{r}\right]$ denotes the operator of taking the coefficient of $y^r$.
It follows that \eqref{eq:two_sums} is equivalent to the following power series having equal coefficients of $y^r$ (which are polynomials in $x$) for all $r\leq 2m$:
\begin{align*}
    F(x,y) := \sum_{l=0}^m\binom{n/g}{l}\binom{2m-n/g}{m-l}\left( 1+xy \right)^{gl}\left( 1+y \right)^{n-gl},\\
    G(x,y) := \sum_{l=0}^m\binom{n/g}{l}\binom{2m-n/g}{m-l}\left( 1+y \right)^{gl}\left( 1+xy \right)^{n-gl}.
\end{align*}
In other words, to prove the lemma, we need to show that 
$F(x, y) \equiv G(x, y) \pmod{y^{2 m+1}}$. 

It is easy to see that
\[\begin{split}
F(x,y) &= [z^m]\ (1+z(1+xy)^g)^{n/g} (1+z(1+y)^g)^{2m-n/g} (1+y)^{n-mg}\\
&= [z^m]\ \left(\frac{(1+y)^g+z(1+xy)^g}{1+z}\right)^{n/g} (1+z)^{2m},
\end{split}\]
and
\[\begin{split}
G(x,y) &= [z^m]\ (1+z(1+y)^g)^{n/g} (1+z(1+xy)^g)^{2m-n/g} (1+xy)^{n-mg} \\
&=[z^m]\ \left(\frac{(1+xy)^g+z(1+y)^g}{1+z}\right)^{n/g} (1+z)^{2m}.
\end{split}\]
Introducing $A:=(1+y)^g$ and $B:=(1+x y)^g$ for the sake of simplicity, we get that
\begin{align}
    F(x, y) &= \left[ z^m \right]\ \left( \frac{A+zB}{1+z} \right)^{n/g}(1+z)^{2m}, \label{eq:F(x,y)}\\
    G(x, y) &= \left[ z^m \right]\ \left( \frac{B+zA}{1+z} \right)^{n/g}(1+z)^{2m}.\notag
\end{align}

\new{Let us define a function $g(z)=\frac{z}{(1+z)^2}$ and introduce a new variable $w=g(z)$. Note that $
z=f(w):=\frac{1-2w-\sqrt{1-4w}}{2w}$, and thus $f(g(z))\equiv 1$. 
Let us write \eqref{eq:F(x,y)} using our new variable. For this, we introduce a function $H(w)$ obtained from the right-hand side of \eqref{eq:F(x,y)} upon the variable change:
\begin{align}
    H(w)=H(g(z)) &:=\Big(\frac{A+B}{2}+\frac{A-B}{2} \sqrt{1-4 w}\Big)^{n / g}\left(\frac{1-\sqrt{1-4w}}{2w}\right)^{2m} \label{eq:H(w)}\\
      &=\Big(\frac{A+zB}{1+z}\Big)^{n/g}(1+z)^{2m}.\notag
\end{align}
Let
\begin{align*}
    \Phi(w) :=\frac w{f(w)}= \frac{2w^2}{1-2w-\sqrt{1-4w}}.
\end{align*}
Our plan is to express $F(x,y)$ using $H(w)$. This task is resolved by the B\"urmann--Lagrange 
lemma \cite[p.733]{flajolet2009analytic} which says that
\begin{align*}
F(x,y)=[z^m]H(g(z)) = [w^m]\{H(w)\Phi(w)^{m-1}( \Phi(w) - w\Phi^\prime(w))\}.
\end{align*}
Substituting $\Phi$ and simplifying, we further obtain}
   $$
   F(x,y) = [w^m]\ \left( \frac{A+B}2 + \frac{A-B}2\sqrt{1-4w} \right)^{n/g}\frac1{\sqrt{1-4w}}.
   $$
Noticing that $A-B$ is a multiple of $y$, we can expand the last formula modulo $y^{2m+1}$ as follows:
\begin{align*}
    F(x,y) = (-4)^m\Big(\frac{A+B}{2}\Big)^{n/g}\sum_{j=0}^{2m}\binom{n/g}{j}\left( \frac{A-B}{A+B} \right)^j \binom{j/2-1/2}{m} + O(y^{2m+1}).
\end{align*}
Note that the expression for $G(x, y)$ can be obtained by exchanging $A$ and $B$. Observe that the terms with odd $j$ are zero (since $\binom{j/2-1/2}{m}=0$), while for even $j$, the corresponding terms in $F(x,y)$ and $G(x,y)$ coincide. This completes the proof of Lemma~\ref{lemma:C1}. 
\end{proof}

\begin{proof}[Proof of Lemma~\ref{lemma:C2}] Straightforward by Zeilberger's ``creative telescoping''. Let $ F(m,l) :=\frac{\binom{m}{l}}{\binom{2x-l}{m+1}} $ and let $ s(m) = \sum_lF(m,l) $, where the summation can be extended to all $ l\in \mathbb{Z} .$  
First notice that
\begin{align}\label{LemmaEqn2}
    2(m+2)F(m,l) - 2(x-m-1)F(m+1,l) = G(m,l+1)-G(m,l),
\end{align}
where
\begin{align*}
    G(m,l):=F(m,l)\frac{l(m+2)}{m-l+1}.
\end{align*}
To see this, divide both sides of (\ref{LemmaEqn2}) by $ F(m,l) $ and use
\begin{align*}
    \frac{F(m,l+1)}{F(m,l)}=\frac{(m-l)(2x-l)}{(2x-l-m-1)(l+1)}, \quad\quad \frac{F(m+1,l)}{F(m,l)} = \frac{(m+2)(m+1)}{(m-l+1)(2x-l-m-1)},
\end{align*}
obtaining the same expression on both sides. Now, sum (\ref{LemmaEqn2}) on $l$ to obtain
\begin{align*}
    2(m+2)s(m) - 2(x-m-1)s(m+1)=0,
\end{align*}
or
\begin{align*}
    s(m+1) = s(m)\frac{m+2}{x-m-1} = s(0)\prod_{j=0}^m\frac{m+2-j}{x-m-1+j} = \frac{1}{2x}\prod_{j=0}^{m}\frac{j+2}{x-1-j}.
\end{align*}
Thus,
\begin{align*}
    s(m) = \frac{1}{2x}\frac{(m+1)!}{(x-m)(x-1)_{(m-1)}} = \frac{m+1}{2(x-m)}\frac{1}{\frac{x}{m}\binom{x-1}{m-1}},
\end{align*}
which is the same as \eqref{eq:Z}. 
\end{proof}

\bibliographystyle{abbrvurl}
\bibliography{PI}

\end{document}